\documentclass[11pt]{article} 

\usepackage[ruled]{algorithm2e}

\SetAlFnt{\small}
\SetAlCapFnt{\small}
\SetAlCapNameFnt{\small}
\SetAlCapHSkip{0pt}
\IncMargin{-\parindent}

\usepackage{amssymb, amsmath, amsthm}
\usepackage{epsfig,graphicx}
\usepackage{fullpage}

\newtheorem{theorem}{Theorem}[section]
\newtheorem{lemma}[theorem]{Lemma}
\newtheorem{corollary}[theorem]{Corollary}

\newtheorem{definition}[theorem]{Definition}
\newtheorem{example}[theorem]{Example}
\newtheorem{observation}[theorem]{Observation}

\newtheorem{proposition}[theorem]{Proposition}
\newtheorem{claim}[theorem]{Claim}

\newcommand {\ignore} [1] {}
\newcommand {\Fb} {\overline{F}}

\begin{document}



\title{Bertrand Networks} 

\author{Moshe Babaioff\thanks{Microsoft Research SVC, \texttt{moshe@microsoft.com}}
\and
Brendan Lucier\thanks{Microsoft Research New England, \texttt{brlucier@microsoft.com}}
\and
Noam Nisan\thanks{Hebrew University and Microsoft Research SVC, \texttt{noam@cs.huji.ac.il}}
}

\date{}

\maketitle

\begin{abstract}
We study scenarios where multiple sellers of a homogeneous good compete on prices, where each seller can only sell to some subset of the buyers.
Crucially, sellers cannot price-discriminate between buyers.   We model the structure of the competition by a graph (or hyper-graph), with nodes representing the sellers and edges representing populations of buyers.
We study equilibria in the game between the sellers, prove that they always exist, and present various structural, quantitative, and computational results about them.  We also analyze the equilibria completely for a few cases.  Many questions are left open.
\end{abstract}

\section{Introduction}
Competition is known to reduce prices and decrease sellers' profits.  The simplest model to this effect
contrasts a seller with a captive market to two competing sellers, where the competition is on
price alone, a model known as Bertrand competition.  While the seller with a captive market would sell at the ``monopoly price'' and make a profit,
the only equilibrium that the
competing sellers may reach is one where they charge the marginal cost, extracting no profit if they have identical constant marginal costs.
There is of course much work in the economic literature that deals with various variants of this model as well as
with alternative assumptions about the competition (e.g. Cournot competition.)

In this paper we study scenarios in which 
parts of the market are shared between sellers and other parts are captive.  We model the structure of
sharing in the market as a hyper-graph where the vertices are sellers, and each hyper-edge represents a market segment, henceforth just market,
that is shared by these sellers.  The sellers each announce a single price (i.e., price discrimination is impossible\footnote{Where price discrimination
is possible, the seller simply optimizes in each market separately.})
and every buyer buys from the lowest price seller that has access to his market (with some tie breaking rule).
The idea is that sellers must balance between
competition in each of the markets they compete in and their captive market, and this trade-off will in turn affect those competing with them.
In this paper we wish to study how the structure of this graph affects
prices and profits of the different sellers.

One may think of many scenarios that are captured -- to a first approximation -- by such a model.  Consider several Internet vendors for some good, where
users do not always compare prices among all vendors but rather different subsets of users do their price-comparisons only between a subset of vendors.
One may also think about geographic limitations to competition where buyers can only buy from a ``close'' vendor.  Another
scenario may involve technology
constraints where buyers must choose between essentially equivalent products, but are limited to buying from the subset of those that
are ``compatible'' with their existing systems or that have a certain ``feature'' that they need.  In all these cases, and many others,
price discrimination would be
quite difficult to do.

Taking a higher-level point of view, this work falls into a more general agenda that attempts
``decomposing'' a global economic situation into a network of local economic interactions and extracting some global economic insights from the
structure of the interaction graph, studied in various models, e.g., in \cite{BlumeEKT07,KLS01,BNP09,KKO04,KM01,B04} and many others.   This agenda is distinct
from agendas that consider network formation or network-structured goods, agendas that have also received much attention,
including in models related to Bertrand competition
\cite{ChawlaN09,ChawlaR08,HPT99} as well as in \cite{Guzman2011} that is this paper's starting point.

\begin{figure}
\centering
\begin{tabular}{c|c}
\includegraphics[scale = 0.62]{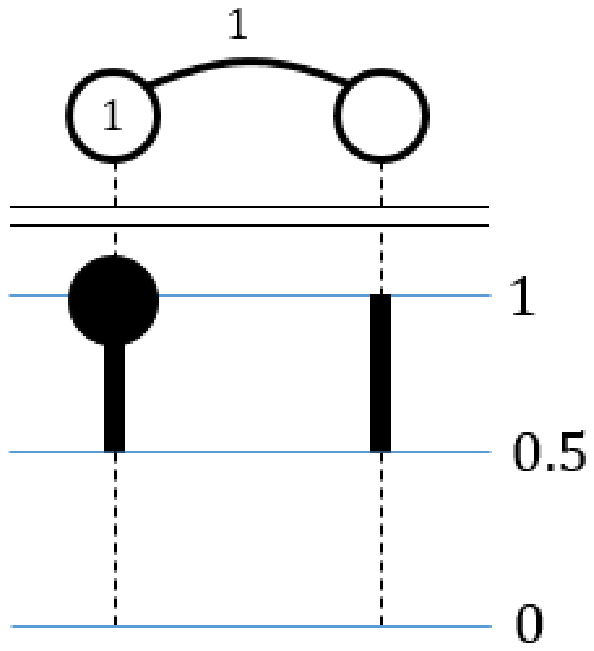} \hspace{1cm} &
\hspace{1cm} \includegraphics[scale = 0.4]{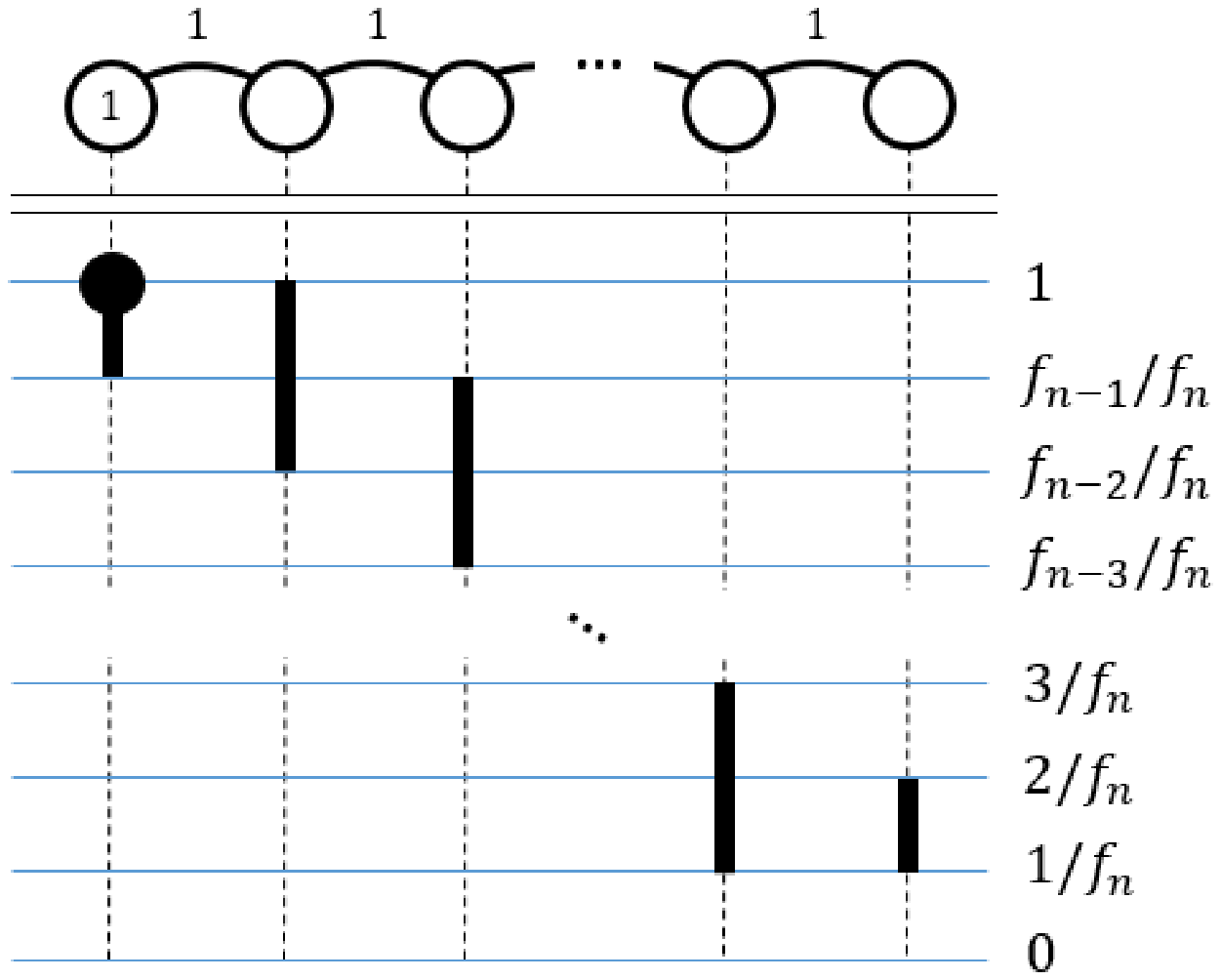} \\
(a) \hspace{1cm} & \hspace{1cm} (b)
\end{tabular}
\caption{An illustration of equilibrium pricing for (a) two sellers with a shared market and one captive market, and (b) a line of $n$ sellers with unit-sized shared markets and a single captive market.  In each depiction, the network is represented at the top.  Captive market sizes are shown within nodes (where a blank node indicates no captive market), and shared market sizes are shown adjacent to their respective edges.  The support of each seller's pricing strategy, as a subset of $[0,1]$, is shown below.  Thick lines indicate the range of prices that fall within each seller's support, and dark circles represent atoms at price $1$. In (b), $f_i$ denotes the $i$th Fibonnacci number, indexed so that $f_0=f_1=1$.}
\label{fig:intro}
\end{figure}

Let us start with the simplest scenario that combines a captive market and a competitive one.  Consider the case of two sellers that share a market,
but where one of the two sellers also has a captive market of the same size as the shared one.  
In our simple scenario both sellers
have zero marginal cost (i.e., for producing the good) and the buyers in each market will all buy from the seller that asked for the lowest price, as long as that price is at most 1.  The
two sellers are thus playing a game, where the strategy of each seller is its requested price which lies in the interval $[0,1]$.
What will the equilibrium look like?  It is easy to verify that no pure equilibrium exists.
However, a mixed equilibrium does exist
and was only recently described in \cite{Guzman2011} (for more general demand and supply curves).
In this unique
equilibrium both sellers randomize their asked price in the range $[0.5,1]$ in
the following way\footnote{One may be somewhat skeptical of the relevance of a mixed Nash equiliribrium with continuous support, however we would like to mention that
we have run simulations 
and found that this mixed continuous support equilibrium was closely approximated by the empirical
distribution of a simple fictitious play in a discretized version of the game.}:
the price of the seller that has the captive market satisfies $Pr[Price < x] = 1-\frac{1}{2x}$ for $0.5 \le x < 1$ and
$Pr[Price = 1] = 0.5$, and that of the seller without a captive market satisfies $Pr[Price < x]=2-\frac{1}{x}$ for all $0.5 \le x \le 1$.  We say that the seller with the captive market has an \emph{atom} at price $1$, meaning that the seller selects price $1$ with positive probability.  See Figure \ref{fig:intro}(a).
It may be somewhat surprising that
the seller with no captive market gets positive utility (of $1/2$) despite having no captive market.
This may be contrasted with what would happen if he also succeeds in gaining access to the other seller's captive market, in which case they would be put
in a classic Bertrand competition and all prices would go down to 0.

Let us continue with another example: a line of $n$ sellers, where each two consecutive ones share a market, and the first one also has a captive market,
with all markets being of the same size.  It turns out that the unique equilibrium has each seller $i$ randomizing his price (according to a specific distribution
that we derive)
in the interval $[f_{n-i+2}/f_n, f_{n-i}/f_n]$ where $f_j$ is the $j$'th Fibonacci number starting with $f_0=f_1=1$
(except for the first seller whose bid is capped at 1, with an atom there).
See Figure \ref{fig:intro}(b).
The equilibrium utilities of the players in this network are given by
$u_i = f_{n-i+1}/f_n = \Theta(\phi^{-i})$, where $\phi$ is the golden ratio.


The paper attempts analyzing what happens in more general situations with {\em multiple sellers and markets} where different sellers are connected
to different subsets of markets.
To focus on the structure of the graph, we keep everything else as simple as
possible, in particular sticking to zero marginal costs as well as to a a demand curve where all buyers are willing to buy the good for at most $1$.\footnote{
This implies that there are no efficiency issues in this model, and our focus is on prices and revenues.}
Furthermore, as the main distinction we wish to capture is that of monopoly as opposed to competition, we focus on the case where
each market is either captive to one seller or shared between exactly two sellers.  This leads us to modeling the
{\em network of sellers and markets} by a graph whose vertices are the sellers
and where each edge corresponds to market that is shared between the two sellers.  Each seller (vertex) $i$ may have a weight $\alpha_i$ indicating the
size of its captive market and each edge $(i,j)$ will have a weight $\beta_{ij}$ indicating the size of the pair's shared market.\footnote{In the more general model of an hyper-graph $\beta_S$ will indicate the size of the market that is shared by the set $S$ of sellers.}
We will analyze Nash equilibria of the game between the sellers.
To begin with, it is not even clear that a Nash equilibrium exists: the game has
a continuum of strategies (the price is a real number) and discontinuous utilities (slightly under-pricing your opponent is very different than slightly overpricing him).
Nevertheless we invoke the results of \cite{Simon-Zame} and show:

\begin{theorem} \label{ithm1}
In every network of sellers and markets there exists a mixed Nash equilibrium.  Moreover, every equilibrium holds for every tie breaking rules.\footnote{
This theorem also holds in the general hyper-graph model.}
\end{theorem}

We then start analyzing the properties of these equilibria.  Extending the well known result about Bertrand competition,
we show 
that if no seller has a captive market then the only equilibrium is the pure one where each seller sells at 0 (his marginal cost) and gets
0 utility.
We observe the following converse:

\begin{theorem}\label{ithm2}
In every connected network of at least two sellers where at least one seller has a captive market, there does not exist any pure Nash equilibrium.
In every mixed-Nash equilibrium of this network
no seller has any atoms,
except perhaps at 1. Moreover, all sellers have their infimum price bounded away from zero, and get strictly positive utility.\footnote{
The fact that lack of captive markets implies zero prices extends to the general hyper-graph model but this theorem does not, nor do the ones below.}
\end{theorem}

We do not have a general algorithm for computing an equilibrium of a given network, however we do show that the problem can be completely reduced
to finding the supports of the sellers' strategies and the set of sellers that have an atom at 1.

\begin{theorem}\label{ithm3}
Given the supports of sellers' strategies, with finitely many boundary points, and the set of sellers that have an
atom at 1, it is possible to explicitly, in polynomial time, compute an equilibrium of the network if such exists.  Generically this equilibrium
is unique for this support and set of sellers with atoms at 1.
\end{theorem}

Generally speaking there may be different equilibria for a network with different
supports of sellers' strategies, with seller's utilities varying between them.
We next embark on an analysis of a set of networks for which we can effectively analyze and prove uniqueness of the equilibrium.

\begin{theorem}\label{ithm4}
Every network of sellers and markets that has a tree structure and a single captive market has an essentially unique equilibrium which is
described explicitly and polynomially computable from the network structure.
\end{theorem}

Our analysis is explicit about what ``essentially unique'' means, completely characterizing the degrees of freedom.  In particular, the utilities
of each seller are the same over all equilibria.  This theorem has two significant limitations: being a tree and having a single captive market.
We provide examples showing that both restrictions are necessary and relaxing either one of them results in multiple equilibria with multiple possible utilities for a seller.   We are able to fully analyze and prove uniqueness of equilibria for an additional case: a ``Star'' where each seller may have a captive market and every
peripheral seller shares a market with the center and all shared markets have the same size.

For general graphs, while equilibria are not necessarily unique, nor are we in general able to characterize them,
we do prove various structural results as well as
quantitative estimates on prices and utilities in every possible equilibria.  We are able to
bound the amount of utility that "flows" from sellers with captive markets to sellers that are ``decoupled'' from them
in each of two senses: (1) distance (2) cut:

\begin{theorem} \label{ithm5}(Informal)
In every non-trivial network and in any equilibrium:
\begin{enumerate}
\item
The utility of every seller is bounded from below by an expression that decreases exponentially in his distance from any captive market.
\item
The utility of every seller is bounded from above by a linear expression in the size of the shared markets in an edge-cut that separates him from all captive markets.
\item
For every seller, as the sizes of all shared markets in an edge-cut that separates him from all captive markets increase to infinity, his utility decreases to 0.
\end{enumerate}
\end{theorem}

\ignore{
\begin{theorem}
Take a non trivial network with $n$ sellers and let $\alpha_{max}=\max_i \alpha_i$ be the size of the largest captive market.  Then
there exist constants $c_1, c_2$ that depend only on the maximum degree of the network as
well as on the maximum ratio between sizes
of markets in the network such that for every seller $i$ we have that
\begin{enumerate}
\item
In every equilibrium, $\alpha_{max} / (c_1)^{d_i} \le u_i \le \alpha_{max} \cdot (c_1)^{d_i}$
where $d_i$ is the distance of $i$ from the seller with captive market $\alpha_{max}$.
\item
Let $E$ be an edge cut that separates $i$ from all captive markets and does not contain edges adjacent to $i$,
and change all market sizes in $E$ to be of size $\eta$ then in every equilibrium of the modified network we have that
$u_i \le {c_2}^n \alpha_{max} \cdot \eta$ and $u_i \le {c_2}^n \alpha_{max} / \eta$.
\end{enumerate}
\end{theorem}
}

Note that our ``line of sellers'' example above shows that the decrease in utility in part 1 of the theorem may indeed be exponential.
Part 3 of the theorem may be surprising, with the intuitive explanation
being that the largeness of the
markets in the cut causes the sellers in these
markets to ``focus'' on them, not letting indirect competition ``spread'' over the cut.

\vspace{0.1in}
\noindent
{\bf Structure of the paper}
\vspace{0.1in}

We start by describing our model in section \ref{sec:model}, and before diving into the body of our analysis, present a few simple
examples in section \ref{sec:simple}.  Our general analysis of the existence, robustness, and properties of equilibria are given
in section \ref{sec:eq} that also proves theorems \ref{ithm1} and \ref{ithm2}.  Section \ref{sec:sketches}
reduces the problem to analysis only at the boundary points, proving theorem \ref{ithm3}.  Section \ref{sec:trees.1.captive} analyzes trees with a single captive
market and proves theorem \ref{ithm4} and section \ref{sec:star} analyzes the star network.  Finally, section \ref{util} analyzes utilities in general
networks, proving a formal version of theorem \ref{ithm5}.  Many open problems remain, and we sketch some of them in our concluding section \ref{sec:conclude}.

\ignore{ 
\subsection{Related Work}
MOSHE: some of the papers we should probably cite include:

Price Competition on Network~\cite{Guzman2011}.

Bertrand competition in networks~\cite{ChawlaNR09}.

Trading networks with price-setting agents~\cite{BlumeEKT07}
}

\section{Model}\label{sec:model}
In a general network economy ({\em network} for short)
there are $n\geq 2$ {\em sellers} and a collection of disjoint buyer populations which we call {\em markets}.
All sellers sell the same type of good, and each seller is associated with a supply curve which specifies how many units the seller can sell at any given price.
Each market has access to some of the sellers, possibly not to all of them.  Each market is associated with a demand curve, specifying how many units the population would buy at a given price.

We will focus on the following subclass of networks.  First, we assume that all buyers in a market are willing to pay up to $1$ per unit but no more.
Also, we assume that each seller has a marginal cost of $0$ for producing the good and is able to supply any quantity.
Each buyer
will purchase a full unit of the good from whichever accessible seller has the lowest price.
Finally, we assume that each market has access to at most two sellers.


As each market has access to at most two sellers,
it is natural to represent a network by a graph as follows. Each seller is represented by a node in the graph.
If a market has  access to only a single seller, we say that this market is {\em captive}.
We write $\alpha_i$ for the {\em size} of the captive market of seller $i$, where $\alpha_i = 0$ if
seller $i$ has no captive market.
Note that we assume without
loss of generality that each seller has at most one captive market, since having two or more is equivalent to
having one with the combined size.
We write $\vec{\alpha}=(\alpha_1,\alpha_2,\ldots,\alpha_n)$.
If a market has access to two sellers $i$ and $j$, we represent that market by an edge from node $i$ to node $j$, and use $\beta_{i,j}$ to denote
the size of that market.
We use $N(i)$ to denote the set of sellers that share a market with seller $i$, and write $\beta_i=\sum_{j\in N(i)} \beta_{i,j}$.
See Figure \ref{fig:simple} for an illustration.

A network defines the following pricing game between the sellers.  
Each seller needs to offer a price per unit of the good.
Each edge (market) buys from the incident node (seller) that offers the lowest price.
A captive market always buys from its associated node.
Formally one needs to specify a tie breaking rule for the case of a tie, but we will later show that ties never occur in equilibrium
(see Section~\ref{sec:eq}), so from that point on will usually omit tie-breaking considerations from our discussion and notation.
Note that each seller offers the {\em same} price to all available (i.e.\ incident) markets (edges).
Sellers offer prices simultaneously, and can use randomization to determine prices.
We assume that  all sellers are risk neutral.

Consider the case that each seller $j$ offers price $x_j$,
the utility of seller $i$ with price $x_i$ in this case is
$$u_i(x_1,x_2,\ldots,x_n) = x_i\left(\alpha_i + \sum_{j\in N(i)} \beta_{i,j} \cdot \chi_{x_i<x_j} \right)$$
Here $\chi_{x_i<x_j}$ is an indicator taking value $1$ if $x_i<x_j$, and $0$ otherwise;  formally, this models $i$ loosing in case of a tie.
As mentioned above, changing the tie breaking rule will result in exactly the same equilibria.

A mixed strategy of seller $i$ can be represented by a CDF $F_i$ with support $S_i$.
The support of $F_i$ is (w.l.o.g) contained in $[0,1]$ (as buyers are not willing to pay more than 1 per item).
We use $F_i^-(x)=\sup_{y<x} F_i(y)$ to denote the probability that $i$ puts strictly below $x$.
If $A_i(x)= F_i(x)- F_i^-(x) >0$ we say that $F_i$ has an {\em atom} at $x$ and $A_i(x)$ is its size.
We denote the probability that $i$ places on at least $x$ by $\Fb_i(x) = 1-F^-_i(x)$.
A point $x$ is a {\em boundary (transition) point} for seller $i$ if every open interval containing $x$ intersects $S_i$ but is not contained in $S_i$.
Note that if $S_i$ is a collection of intervals, then the set of boundary points is precisely the set of endpoints of these intervals.
We use $\sup_i$ and $\inf_i$ to denote the supremum and infimum of $S_i$. 

We can now define $u_i(x, F_{-i})$, the utility of (risk-neutral) seller $i$ when declaring price
$x\in [0,1]$, when the other sellers price according to $F_{-i}$:
\begin{equation}
\label{eq:utility}
u_i(x, F_{-i})=x\left(\alpha_i+\sum_{j\in N(i)} \beta_{i,j}\left(1-F_j(x)\right)\right)
\end{equation}

As mentioned, in Section~\ref{sec:eq} we show that ties do not matter and that no two neighboring sellers can both have an atom at $1$.
From that point on, when considering an equilibrium,
it would be notationally convenient to slightly deviate from the formula above that corresponds to $i$ loosing the tie with $j$ as we formally defined.
Instead, for a seller $i$ that has a neighbor $j$ with an atom at 1
we will 
replace the above by 
the formula that corresponds to $i$ winning the tie at 1 against $j$ and define:
$$u_i(1,F_{-i}) = \left(\alpha_i+\sum_{j\in N(i)} \beta_{i,j}\left(1-F^-_j(1)\right)\right).$$
This is notationally convenient as it maintains $u_i(1,F_{-i})=\lim_{x_i \rightarrow 1} u_i(x_i,F_{-i})$ so
in many arguments this avoids the extra notation of taking limits as $x_i$ approaches $1$.
In particular, this notation is useful as it allows us to think of every price in the support $S_i$ as being optimal for $i$. This is trivially true for every point in which the utility of $i$ is continuous. As atoms only happen at $1$, the price of $1$ is the only possible point of discontinuity. With this definition of $u_i(1,F_{-i})$ the utility of seller $i$ with supremum price of $1$ is also optimal at $1$.
We use $u_i$ to denote the equilibrium utility of seller $i$.  
Additionally, when $F_{-i}$ is clear from context we will abuse notation and write $u_i(x)= u_i(x, F_{-i})$.



A network consists of a graph and market sizes.
We say that a network is \emph{non-trivial} if it is connected, has at least two sellers, and has at least one captive market.
For most of the paper we will focus on non-trivial networks.\footnote{Indeed, for disconnected graphs our results will hold for each component separately, and the degenerate case of no captive markets is solved in Theorem~\ref{thm:pure-eq} and thus is irrelevant to any later parts of the paper.}

\section{Simple Examples}
\label{sec:simple}

We begin by building some intuition for our pricing game by describing a few simple examples.  This intuition will be helpful when describing  general properties of equilibria in Section \ref{sec:eq} and the structure of equilibria in Section \ref{sec:sketches}.

The simplest network is a single seller that is a monopolist over a single market.  In this case he will price the item at 1 and extract all surplus.
Another simple network is the case of two sellers with no captive markets who share a single market; this is precisely a Bertrand competition (with marginal cost of 0).
In this example the 
unique equilibrium is for both sellers to price the item at $0$ (regardless of the size of the shared market), and all surplus goes to the buyers.

We now consider two more interesting examples with non-trivial networks.

\begin{figure}
\centering
\begin{tabular}{c|c}
\includegraphics[scale = 0.62]{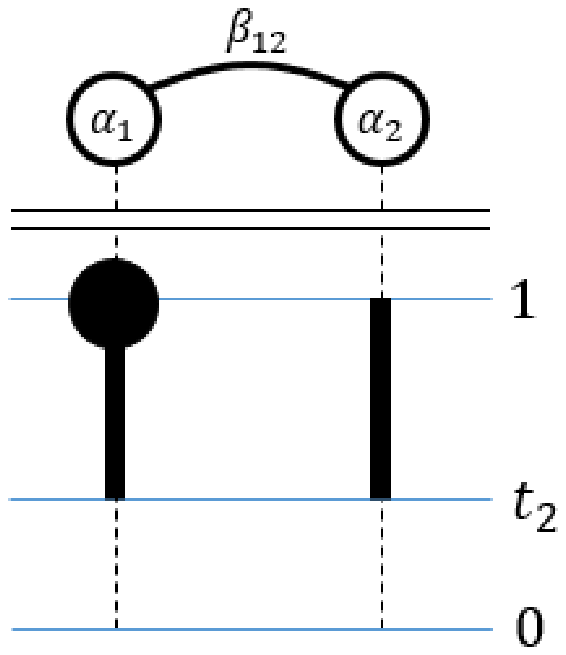} \hspace{1cm} &
\hspace{1cm} \includegraphics[scale = 0.4]{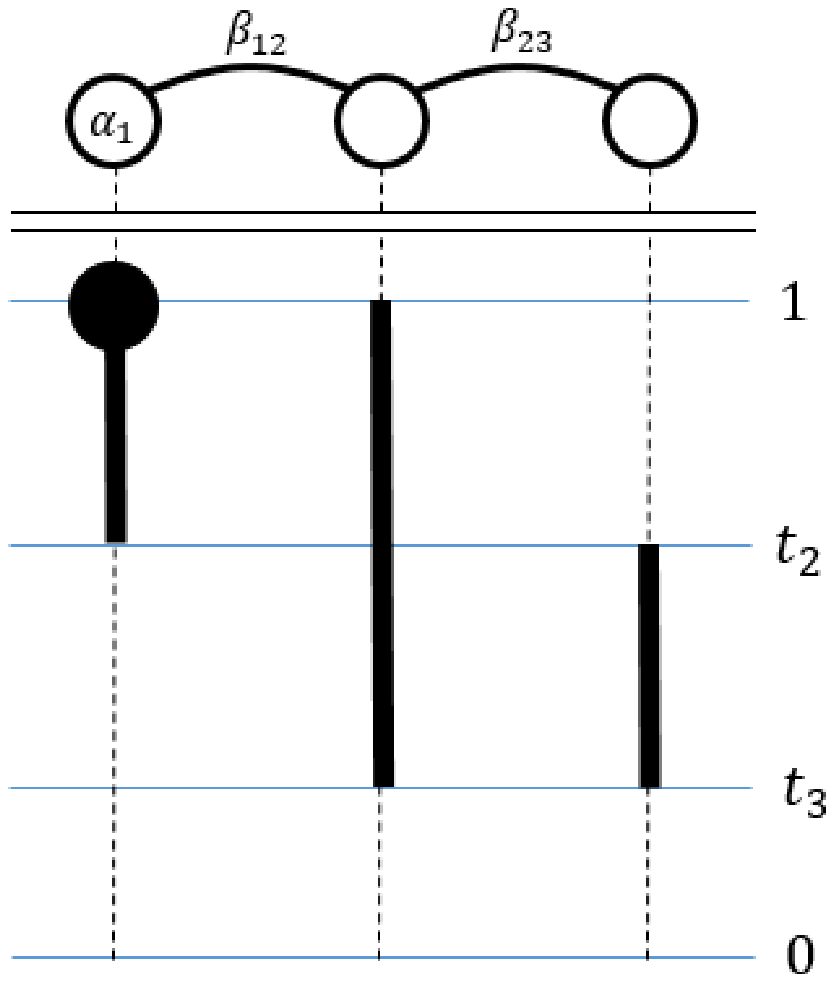} \\
(a) \hspace{1cm} & \hspace{1cm} (b)
\end{tabular}
\caption{An illustration of equilibria for (a) a line of two sellers with $\alpha_1 > \alpha_2$ and (b) a line of three sellers with a single captive market $\alpha_1$.  In each depiction, the network is represented at the top.  Captive market sizes are shown within nodes (where a blank node indicates no captive market), and shared market sizes are shown adjacent to their respective edges.  The support of each seller's pricing strategy, as a subset of $[0,1]$, is shown below.  Thick lines indicate the range of prices that fall within each seller's support, and dark circles represent atoms at price $1$.}
\label{fig:simple}
\end{figure}

\begin{example}[General case of $2$ Sellers]
\label{example:line-2}
Consider two sellers that share a market, where additionally each seller has his own captive market. The captive markets have sizes $\alpha_1\geq \alpha_2>0$ and the shared market has size $\beta_{1,2}>0$.  Theorem
\ref{thm:necessary-eq} will imply that, in the unique equilibrium, the support of each seller's strategy is some interval $[t_2,1]$ where $0 < t_2 < 1$, and moreover seller $2$ has no atom at $1$.  See Figure \ref{fig:simple}(a).  It holds that $\Fb_1(t_2)=\Fb_2(t_2)=1$ and $\Fb_2(1)=0$.  If seller $1$ sets price $1$, he will sell only to his captive market and lose the shared market to seller $2$ with probability $1$.  On the other hand, if he sets price $t_2$, he will win the shared market with probability $1$.  Since prices $1$ and $t_2$ are both in the support of $F_1$, it must therefore hold that $\alpha_1=u_1(1)= u_1(t_2)= t_2(\alpha_1+\beta_{1,2})$, and thus $t_2 = \frac{\alpha_1}{\alpha_1+\beta_{1,2}}$.
Applying similar reasoning to seller $2$, we have $\alpha_2+\beta_{1,2}\Fb_1(1)=u_2(1)= u_2(t_2)= t_2(\alpha_2+\beta_{1,2})$, thus the size of the atom of seller $1$ at $1$ is
$\Fb_1(1)= \frac{t_2(\alpha_2+\beta_{1,2})- \alpha_2}{\beta_{1,2}}$. Note that seller $1$ has no atom if and only if the sellers are symmetric ($\alpha_1=\alpha_2$).

We can now explicitly find the CDFs, using the fact that sellers must be indifferent within their supports.
For every  $x\in [t_2,1]$ it holds that $\alpha_1=u_1(1)= u_1(x)= x(\alpha_1+\beta_{1,2}\Fb_2(x))$, and thus
$\Fb_2(x)= \frac{\alpha_1}{\beta_{1,2}}\cdot \left(\frac{1}{x}-1\right)$.
It also holds that
$ t_2(\alpha_2+\beta_{1,2})= u_2(t_2)= u_2(x)= x(\alpha_2+\beta_{1,2}\Fb_1(x))$, and thus
$\Fb_1(x)= \frac{1}{\beta_{1,2}}\left(\frac{t_2(\alpha_2+\beta_{1,2})}{x}- \alpha_2\right)$.
Note that the seller with the larger captive market gains nothing from the shared market (his utility is $\alpha_1$), while the other seller gains more than $\alpha_2$ when the sellers are asymmetric.
\end{example}

\begin{example}[$3$ Sellers in a line with $1$ captive market]
\label{example:line-3-single-captive}
In this example, seller $1$ has a captive market of size $\alpha_1>0$ and shares a market of size $\beta_{1,2}>0$ with seller $2$. Seller $2$ shares a market of size $\beta_{2,3}>0$ with seller $3$. Neither seller $2$ nor $3$ has a captive market.
As we prove in Section~\ref{sec:tree-line}, the unique equilibrium has the following form.  Seller $1$ has an atom of size $A =\Fb_1(1)$ at $1$.
For some $1=t_1>t_2>t_3>0$, the support of seller $1$ is $[t_2,t_1]$, the support of seller 2 is $[t_3,t_1]$, and the support of seller $3$ is $[t_3,t_2]$.  
See Figure \ref{fig:simple}(b) for a ``sketch'' of this equilibrium structure.

Given the form of the equilibrium, it is possible to solve for the values of $t_2, t_3$ and $\Fb_1(1)$ and $\Fb_2(t_2)$ 
in a method similar to Example \ref{example:line-2}.
This turns out to give $\Fb_2(t_2)=\frac{\beta_{1,2}}{\beta_{1,2}+\beta_{2,3}}$, $t_2 = \Fb_1(1) = \frac{\alpha_1}{\alpha_1+\beta_{1,2}\Fb_2(t_2)}$,   
and $t_3=\frac{t_2\beta_{1,2}}{\beta_{1,2}+\beta_{2,3}}$.   We now have the values of $F_i(t_j)$ for all $i$ and $j$,
and, similarly to Example \ref{example:line-2} we can deduce the 
full form of the CDFs $(F_1,F_2,F_3)$ that will be the piece-wise linear in $x^{-1}$ functions that agree with these values.
The general methodology of finding the equilibrium from the ``sketch'' is described in section \ref{sec:sketches} with the details
for general line networks in Section \ref{sec:tree-line}. 
\end{example}


\section{Equilibrium Analysis}
\label{sec:eq}

In this section we study the existence and properties of equilibria in pricing networks.  We first establish that ties occur with probability $0$ in any equilibrium.  We then show that the non-occurence of ties implies that an equilibrium always exists.  Finally, we describe some general properties of every equilibrium.

\subsection{Tie Breaking}
\label{sec:tie-breaking}

We first show that any valid tie breaking rule results in the same set of equilibrium.  Moreover, in any equilibrium, the utility of each seller is independent of the tie-breaking rule.

A {\em valid tie breaking rule} specifies for every two sellers $i$ and $j$ that share a market, and every price vector $\vec{p}=(p_1,p_2,\ldots,p_n)$ with $p_i=p_j$, the fraction of the market that buys from $i$ and $j$ respectively: $f_{i,j}(\vec{p})\geq 0$ and $f_{j,i}(\vec{p})\geq 0$ that $i$ and $j$ respectively (where $f_{i,j}(\vec{p})+ f_{j,i}(\vec{p})\leq 1$). 
To determine the impact of tie-breaking, let us revisit the definition of seller utilities in case of a tie.
Consider the case that each seller $j$ offers price $p_j$.  The utility of seller $i$ is then
$$u_i(p_1,p_2,\ldots,p_n) = p_i\left(\alpha_i + \sum_{j\in N(i)} \beta_{i,j} \cdot \chi_{i,j}(\vec{p}) \right)$$
where $\chi_{i,j}(\vec{p})$ is the fraction of the market shared by $i$ and $j$ for which $i$ sells.
That fraction is $1$ if $p_i<p_j$, and is $f_{i,j}(\vec{p})$ if $p_i=p_j$.

We would like to compute the utility $u_i(x, F_{-i})$ when seller $i$ uses price $p_1$ and the others sample according to $F_{-i}$.
Define
$E_{i,j}(p_i,F_{-i}) =
E_{\vec{p}_{-i}\sim F_{-i}} [f_{i,j}(\vec{p}) | p_j=p_i]
$
As sellers are risk neutral, the utility obtained by seller $i$ when selecting price $x$, assuming others set prices according to $F_{-i}$, is
$$u_i(x, F_{-i})=x\left(\alpha_i+\sum_{j\in N(i)} \beta_{i,j}\left(1-F_j(x) + A_j(x) \cdot E_{i,j}(x,F_{-i})\right)\right).$$

We can now show that tie breaking has no impact on the equilibria of the game.

\begin{theorem}
\label{thm:tie-breaking}
Fix any network.
If a profile of strategies is an equilibrium with some valid tie breaking rule, then that profile is an equilibrium for any other valid tie breaking rule. Moreover, in each such equilibrium, the utility of each seller is independent of the tie breaking rule.
\end{theorem}
\begin{proof}
As ties at  price $0$ do not influence seller utilities, it is enough to show that ties at positive prices have measure zero in any equilibrium.
To prove this it is enough to prove the following lemma.  

\begin{lemma}
\label{lemma:no-same-atom}
Fix any valid tie breaking rule.
In any network and any equilibrium, no two sellers who share a market both have an atom at the same positive price.
\end{lemma}
\begin{proof}
The lemma follows from the fact that for one seller $i$, a slight decrease in the price will allow $i$ to win over the atom for sure (instead of just a fraction of the time due to tie breaking) and increase his utility. We next formalize this claim.


Assume that $i$ and $j$ share a market and both have an atom at $x>0$. We assume without loss of generality that $E_{i,j}(x,F_{-i})<1$ (otherwise replace $i$ and $j$. Note that $E_{i,j}(x,F_{-i}) + E_{j,i}(x,F_{-i})\leq 1$).

Note that $x$ is an optimal price for seller $i$.
Assume that a seller $j$ that shares a market with $i$ and has an atom of size $A_j(x)>0$.
We show that there is a price $y<x$ with $u_i(y)>u_i(x)$, contradicting the assumption that $x$ is optimal for $i$.
Indeed, for $y<x$

$$ u_i(y)=y\left(\alpha_i+\sum_{j\in N(i)} \beta_{i,j}\left(1-F_j(y) +
A_j(y) \cdot E_{i,j}(y,F_{-i})\right)\right)\geq
y\left(\alpha_i+\sum_{j\in N(i)} \beta_{i,j}\left(1-F_j(y)\right)\right)$$

thus
$$\lim_{y\rightarrow x, y<x} u_i(y)\geq  \lim_{y\rightarrow x, y<x} y\left(\alpha_i+\sum_{j\in N(i)} \beta_{i,j}\left(1-F_j(y)\right)\right)=
x\left(\alpha_i+\sum_{j\in N(i)} \beta_{i,j}\left(1-F^-_j(x)\right)\right) = $$

$$ x\left(\alpha_i+\sum_{j\in N(i)} \beta_{i,j}\left(1-F_j(x) + A_j(x)\right)\right) >
x\left(\alpha_i+\sum_{j\in N(i)} \beta_{i,j}\left(1-F_j(x) + A_j(x) \cdot E_{i,j}(x,F_{-i})\right)\right) = u_i(x)
$$
where the strict inequality follows from the existence of a seller $j$ that is a neighbor of $i$ for which it holds that $j$ has an atom at $x$ ($A_j(x)>0$) and $E_{i,j}(x,F_{-i})<1$.
\end{proof}

This concludes the proof of the theorem.
\end{proof}


We note that Lemma~\ref{lemma:no-same-atom} implies that the utility of every seller at every point smaller than $1$ is continuous in his price.  Thus any price in $S_i$, including the boundary of $S_i$, 
is optimal for the seller.

\subsection{Existence of equilibrium}

We show that a mixed equilibrium is guaranteed to exist for any network.
This is a non-trivial claim, since the strategy space is infinite and utilities are discontinuous.
\begin{theorem}
\label{thm:equil.exists}
In any network there exists a mixed equilibrium.
\end{theorem}
\begin{proof}
The existence of a mixed equilibria in our game follows from the general results of \cite{Simon-Zame}.  They consider
general games where the strategy sets are compact metric spaces and the utility functions are only defined to be continuous on a dense subset
of the space of strategy profiles.  Their main motivation is scenarios where the utilities are continuous everywhere except at sparse
``tie points'' in which some discontinuity occurs. This is exactly the case we have in our setting where the strategy set of a seller is
the interval $[0,1]$ and the utility of every seller is continuous (linear in his own price) everywhere except at points where
his price equals that of another seller, in which case a discontinuous jump in utility occurs.  To place our setting into their
formalism we simply consider the subspace of strategy profiles that have no ties, which is a dense subset, and over this subset the utilities in our game are continuous.

The main result of \cite{Simon-Zame} is that as long as we allow our equilibrium to {\em endogenously} choose ``tie-breaking'' utilities
for the strategy profiles
that lie outside the dense subset then a mixed Nash equilibrium exists.  Specifically, the endogenously chosen profile of utilities lies in the convex
hull of the closure of the graph of utilities in the dense subset over which the utility function was exogenously defined and continuous.
In our setting, at a point with a tie between sellers and $i$ and $j$ the endogenously chosen utilities for $i$ and $j$ will be some convex combination
of the utility when $i$ wins the market in case of tie and when $j$ does so.  That corresponds to
each of the two sellers winning some fraction of the market in a tie, with the
sum of the fractions being exactly 1.

At this point we can invoke the fact that, for our games, the tie breaking rule does not matter as discussed in Section \ref{sec:tie-breaking}: for the endogenously-chosen tie breaking rule,
a mixed Nash equilibrium exists by the results of \cite{Simon-Zame}.
This tie breaking rule certainly falls into the family of tie-breaking rules considered in Section \ref{sec:tie-breaking}.  Therefore,
for any other tie-breaking rule in this family, the same profile of mixed strategies is still a mixed-Nash equilibrium.
\end{proof}

Theorem \ref{thm:equil.exists} shows that an equilibrium exists, but is it unique? In the example presented in Section~\ref{example:non-eq-utils} we show that there may exist multiple equilibria.  Moreover, these equilibria are truly distinct from the perspective of the sellers, in the sense that they are not utility-equivalent (i.e.\ some sellers' utilities differ between the equilibria).

Is it possible that a \emph{pure} equilibrium exists?
We observe that when at least one seller has a captive market, a pure equilibrium never exists.
Recall that a non-trivial network is connected, has at least two sellers, and has at least one captive market.
\begin{observation}
\label{obs:no-pure}
If a network is non-trivial then
there does not exist a pure equilibrium (that is, in any equilibrium at least one seller uses a mixed strategy).
\end{observation}
\begin{proof}
Assume that a pure equilibrium exists.  Note that not all sellers can choose price $0$, as a seller with a captive market would generate positive utility by selecting a positive price.
We further claim that \emph{no} seller can choose price $0$.  Indeed, if some seller chooses price $0$, then there exists a seller that chooses price $0$ and that has a neighbor $j$ that chooses positive price $p_j > 0$.  In this case, this seller with price $0$ receives utility $0$, but would receive positive utility (from the market shared with $j$) if he chose price $p_j / 2$.  This contradicts the equilibrium assumption, and hence no seller chooses price $0$.

Let $i$ be a seller with minimal price $p_i > 0$.
By Lemma~\ref{lemma:no-same-atom} none of his neighbors price at $p_i$. As $i$ has finitely many neighbors and they all price using a pure strategy, there is an $\epsilon>0$ such that if $i$ increases his price by $\epsilon$ he sells to exactly the same set of buyers for a higher price, increasing his utility. This contradicts the equilibrium assumption.
\end{proof}

Observation \ref{obs:no-pure} does not consider networks in which no seller has a captive market. For networks with no captive markets we show that the unique equilibrium is a pure equilibrium in which every seller selects price $0$.
\begin{theorem}
\label{thm:pure-eq}
Consider any connected network with at least two sellers.
If no seller has a captive market then the unique equilibrium is for all sellers to price the good at $0$ ($F_i(0)=1$ for all $i$).  In this equilibrium every seller has zero utility.
\end{theorem}
\begin{proof}
Consider any equilibrium (either pure or randomized) and assume that not all sellers always price the good at $0$. This means that for some seller $i$ it holds that $F_i(0)<1$ and that $\sup_i>0$. This implies that any seller $j$ that is neighbor of $i$ has positive utility, and thus positive infimum price. By Observation~\ref{obs:inf-price-prop} every seller has a positive infimum price and positive utility in equilibrium. Consider a seller $j$ with maximal supremum price, breaking ties in favor of a seller that has an atom at that price. This means that if $j$ does not have an atom at $\sup_j$, none of his neighbors has an atom.
Moreover, if $j$ has an atom at $\sup_j$ it is still true that none of his neighbors has an atom at this price due to Lemma~\ref{lemma:no-same-atom}.
In any case $\sup_j$ is in the support of $j$ and when pricing at $\sup_j$ seller $j$ make no sell in any of his non-captive markets.
That seller has no captive market and thus he never sells and has zero utility, a contradiction.
\end{proof}

Motivated by Theorem \ref{thm:pure-eq}, we consider only non-trivial networks in later sections.

\subsection{Properties of Equilibria of Non-Trivial Networks}

We next present some properties that every equilibrium in a non-trivial network must satisfy.

\begin{theorem}
\label{thm:necessary-eq}
Fix any non-trivial network and equilibrium. The following holds:
\begin{enumerate}
\item \label{thm:captive-positive} There exists some positive $\delta>0$ (independent of the equilibrium) such that 
the support of prices of every seller is contained in $[\delta,1]$. Moreover, 
every seller $i$ has positive utility, and his utility is at least $\alpha_i$.
\item \label{thm:atoms-at-1} If seller $i$ has an atom, that atom must be at $1$, and it must be the case that $i$ has a captive market. None of the neighbors of $i$ has any atoms.
\item \label{thm:nec-1} If seller $i$ has no captive market and none of his neighboring sellers has an atom at $1$, then seller $i$'s supremum price is strictly less than $1$.
\item \label{thm:nec-union-support} For any seller $i$ the support $S_i$ excluding the point $1$
is contained in the union of the supports of 
the neighbors of $i$.
\item \label{thm:nec-local-max} If the supremum of the support of seller $i$ is at least the supremum of the support of all his neighboring sellers then the supremum of his support is $1$.
\item \label{thm:nec-util-alpha} There is at least one seller $i$ with utility $u_i=\alpha_i$. That seller has a captive market ($\alpha_i>0$) and $1 \in S_i$. Any seller with no captive market has no atoms.
\end{enumerate}
\end{theorem}

The proof of the theorem follows from the following sequence of claims and observations.  Our first observation holds for any network.
\begin{observation}
\label{obs:inf-price-prop}
Fix any network and any equilibrium.
If there is at least one seller $i$ with support that has a positive infimum ($\inf_i>0$), then there exists some positive $\delta>0$ such that in any equilibrium the support of prices of every seller is contained in $[\delta,1]$. Moreover, in any equilibrium every seller $i$ has positive utility, and the utility is at least $\alpha_i$.
\end{observation}
\begin{proof}
We show that any seller that has a neighbor with positive infimum price, also have positive infimum price. Indeed, assume that $i$ has $\inf_i>0$ and consider a seller $j$ that shares a market of size $\beta_{i,j}>0$ with $i$.  for small enough $\epsilon>0$, by pricing at $\inf_i-\epsilon>0$ seller $j$ can unsure utility of at least $(\inf_i-\epsilon)(\alpha_j+\beta_{i,j})>0$, thus any price $y$ in the support of $j$ must be at least
$\frac{(\inf_i-\epsilon)(\alpha_j+\beta_{i,j})}{\alpha_j + \beta_j}>0$.

The above claim implies the existence of $\delta>0$ such that in any equilibrium the support of prices of every seller is contained in $[\delta,1]$.
This implies that every seller $i$ has positive utility in equilibrium, as the utility must be at least as high as the utility achieved by pricing at $\delta/2$, which is least $(\alpha_i+\beta_i)\delta/2 >0$.

Finally, observe that in any equilibrium the utility of $i$ is at least $\alpha_i$, as by pricing the good at $1$ seller $i$ gets utility of at least $\alpha_i$ (for any strategies of the others).
\end{proof}

Given the lemma we prove the following corollary which implies Theorem~\ref{thm:necessary-eq} (\ref{thm:captive-positive}).
\begin{corollary}
\label{cor:captive-positive}
Fix any non-trivial network (connected with at least one captive market).
There exists some positive $\delta>0$ such that in any equilibrium the support of prices of every seller is contained in $[\delta,1]$. Moreover, in any equilibrium every seller $i$ has positive utility, and the utility is at least $\alpha_i$.
\end{corollary}
\begin{proof}
Consider some seller $i$ with a captive market of size $\alpha_i>0$.
If $i$ prices at $x$, his utility is at most $x(\alpha_i + \beta_i)$ (recall that $\beta_i$ is the total size of all non-captive markets of $i$), thus
any price in the support of $i$ must be at least $\frac{\alpha_i}{\alpha_i + \beta_i}>0$.
The claim now follows from Observation~\ref{obs:inf-price-prop}.
\end{proof}
Note that this in particular says that the profile in which all sellers post a price of $0$ is not an equilibrium when there is a captive market.

\begin{observation}
\label{obs:neighbor-above}
Fix any non-trivial network and any equilibrium.
If price $z<1$ is in the support of seller $i$ then there exists a neighbor $j$ of $i$ such that for any $x>z$ it holds that $F_j(x)>F_j(z)$.
\end{observation}
\begin{proof}
By Corollary~\ref{cor:captive-positive} seller $i$ has positive utility and thus wins with positive probability with the price of $z$.
There must exist a neighbor $j$ of seller $i$ such that for any $x>z$ it holds that $F_j(x)>F_j(z)$, as otherwise a small enough increase in the price by $i$ will result with higher utility for him (he still wins the same buyers with the same positive probability, but for a higher price).
\end{proof}

The next observation implies Theorem~\ref{thm:necessary-eq} (\ref{thm:atoms-at-1}).
\begin{observation}
\label{obs:atoms-at-1}
Fix any non-trivial network and any equilibrium.
If seller $i$ has an atom, that atom must be at $1$, and it must be the case that $i$ has a captive market.
None of the neighbors of $i$ has any atoms.
\end{observation}
\begin{proof}
Assume in contradiction that in some equilibrium there is a seller $i$ with an atom at some $z<1$, which means that $z$ is in the support.
By Corollary~\ref{cor:captive-positive} seller $i$'s support has positive infimum ($\inf_i>0$), and thus has no atom at $0$, so we can assume that $z>0$.
By Observation~\ref{obs:neighbor-above} there exists a neighbor $j$ of $i$ such that for any $x>z$ it holds that $F_j(x)>F_j(z)$.
This means that $j$ has optimal prices arbitrarily close to $z$ (above $z$).
By Corollary~\ref{cor:captive-positive} $j$ wins with positive probability with any price in his support.
Now, seller $j$ can increase his utility by pricing at $y<z$ that is large enough, as he now also wins over the atom of $i$ but losses arbitrarily small in price (the formal argument is similar to the one presented in Lemma~\ref{lemma:no-same-atom} and is omitted).

Finally, if a seller has an atom at $1$ this means that his utility is $\alpha_i$ (as none of his neighbors has an atom at 1 by
Lemma~\ref{lemma:no-same-atom}). If he has no captive market this means that his utility is zero, in contradiction to
Corollary~\ref{cor:captive-positive}.

None of the neighbors of $i$ has any atoms as any such atom must be at $1$, but that is impossible by Lemma~\ref{lemma:no-same-atom}.
\end{proof}


The next observation implies Theorem~\ref{thm:necessary-eq} (\ref{thm:nec-1}).
\begin{observation}
Fix any non-trivial network.
Consider any seller $i$ that has no captive market and assume that in some equilibrium none of his neighboring sellers has an atom at $1$. Then seller $i$'s supremum price in that equilibrium is strictly less than 1.
\end{observation}
\begin{proof}
Seller $i$ has no captive market. Assume that $\sup_i=1$. As none of his neighbors has an atom at $1$, his utility is continuous at $1$ and is eqaul to $\alpha_i$ which is $0$ as he has no captive market.
But, if $i$ has a 
captive market then he has positive utility by Corollary~\ref{cor:captive-positive}. A Contradiction.
\end{proof}

The next observation implies Theorem~\ref{thm:necessary-eq} (\ref{thm:nec-union-support}).
\begin{observation}
\label{obs:contained-support}
Fix any non-trivial network and any equilibrium.
For any seller $i$ the support $S_i$ excluding the point $1$
is contained in the union of the supports of 
the neighbors of $i$.
\end{observation}
\begin{proof}
By Observation~\ref{obs:atoms-at-1} no seller has any atom, except possibly at $1$.


Assume that the claim is not true, then for some seller $i$ and some prices $1>y>x$ in the support of $i$ it holds that $F_j(x) = F_j(y)$ for every neighbor $j$ of $i$. It is easy to see that in this case the utility of $i$ by price $y$ is strictly larger then his utility by price $x$, contradiction the assumption that $x$ is optimal for $i$ (any point in the support that is not $1$ is optimal).
\end{proof}

The next corollary shows that any local minimum of the infima must be shared by at least two sellers.
\begin{corollary}
\label{cor:local-infimum}
Fix any non-trivial network and any equilibrium.
If the infimum of the support of seller $i$ is at most the infimum of the support of all his neighboring sellers then there is some neighboring seller with the same support infimum.
\end{corollary}

The next observation shows that any local maximum of the suprema is a global maximum.
It implies Theorem~\ref{thm:necessary-eq} (\ref{thm:nec-local-max}).
\begin{observation}
\label{obs:local-supremum}
Fix any non-trivial network
and any equilibrium.
If the supremum of the support of seller $i$ is at least the supremum of the support of all his neighboring sellers
then the supremum of his support is $1$.
\end{observation}
\begin{proof}
If for seller $i$ it holds that $1>sup_i\geq sup_j$  for every $j$ that is a neighbor of $i$ then $i$ has utility zero, as no seller has an atom at a positive price that is less than $1$ (Observation~\ref{obs:atoms-at-1}). This contradict Corollary~\ref{cor:captive-positive} which shows that $i$ must have positive utility in any equilibrium. We conclude that $sup_i=1$.
\end{proof}

The next observation implies Theorem~\ref{thm:necessary-eq} (\ref{thm:nec-util-alpha}).
\begin{corollary}
\label{util_captive}
Fix any non-trivial network and any equilibrium.
There is at least one seller $i$ with utility $u_i=\alpha_i$, that seller has a captive market ($\alpha_i>0$).
\end{corollary}
\begin{proof}
Consider seller $i$ with the maximum supremum price, breaking ties in favor of a seller with an atom.
By Observation\ref{obs:local-supremum} it holds that $\sup_i=1$. None of $i$'s neighbors has an atom at $1$ by Lemma~\ref{lemma:no-same-atom}.
When seller $i$ prices arbitrarily close to $1$ he only wins his captive market, thus his utility is $\alpha_i$.
It must be the case that $\alpha_i>0$, as every seller has positive utility, by Corollary~\ref{cor:captive-positive}.
\end{proof}

\section{Supports and Equilibrium}
\label{sec:sketches}

In general, the definition of a mixed Nash equilibrium requires checking a continuum of equations and inequalities.
In this section we show that the space of potential equilibria -- and the conditions to check -- can be simplified immensely.
An equilibrium \emph{sketch} (defined formally below) describes each seller's support and the set of players that have an atom at $1$. 
We will show that once the sketch of an equilibrium is known, the full specification of an equilibrium with that support can be determined.  
Moreover, one can efficiently decide whether a given sketch corresponds to an equilibrium: it suffices to check the equilibrium conditions at the boundary points of the players' supports.  We also provide conditions under which a sketch \emph{uniquely} determines an equilibrium.

Assume that we are given the support $S_i$ of each CDF $F_i$ for every seller $i$.
Let $B_i$ be the set of boundary points for the support $S_i$, and let $T=\cup_i^n \ B_i$ be the union of all these sets, we call it the set of
{\em boundary points} of $\{ F_i \}_i$.
Let $T_i$ be the set of points in $S_i\cap T$.  That is, $T_i$ is the set of boundary points that are in the support of seller $i$.

We say that support $S_i$ has {\em finite boundary} if $|B_i|$ is finite.
Suppose all sellers have supports with finite boundary.  Then $T$ is finite; write $k = |T|$.
We can then write $T = \{t_1, t_2, \dotsc, t_k\}$, where $1\geq t_1>t_2>\ldots>t_k\geq 0$.
Note that if the CDFs form an equilibrium then $t_1=1$ and $t_k>0$ (by Theorem \ref{thm:necessary-eq}, items \ref{thm:captive-positive} and \ref{thm:nec-util-alpha}).
It will sometimes be convenient to think of the list of points $\tilde{T}$ that also includes the point $t_{k+1}=0$, so we denote
$\tilde{T}= \{t_1,t_2,\ldots,t_k,t_{k+1}\}$.
Additionally, for $j\in \{1,2,\ldots,k\}$ we denote by $R_j$ the set of sellers with support that contains the interval $(t_{j+1},t_j)$.  Note that $R_k$ is empty when $t_k>0$, as happens in any equilibrium.
Finally, $R_0$ specifies the set of sellers that have an atom at $1$.


\begin{definition}
A {\em sketch} (of an equilibrium) specifies for every seller $i$ the support $S_i$ of $F_i$, where all supports have finite boundary.
Additionally, the sketch specifies a set $R_0$ of sellers that should have atoms at $1$. 
An equilibrium {\em satisfies the sketch} if its supports and atoms match those of the sketch.
\end{definition}

A \emph{sketch solution} is a sketch augmented with partial information about a Nash equilibrium, concerning behavior at boundary points.
Recall that for seller $i$ and point $x\in [0,1]$, we denote $\Fb_i(x) = 1-F^-_i(x)$. Since no seller has an 
atom at $x < 1$ in equilibrium, we have $\Fb_i(x) = 1-F_i(x)$ for $x < 1$.  Also, $\Fb_i(1) = 1-F^-_i(1)$ is the size of the atom of $i$ at $1$.  

\begin{definition}
A {\em sketch solution} (of an equilibrium) specifies a sketch and, additionally, it defines values $\Fb_r(t)$
for every seller $r$ and point $t$ in the set $T$ of boundary points of the sketch.  These values must satisfy the following linear program (LP1) in the variables $\{u_i\}_{i\in [n]}$, $\{\Fb_r(t)\}_{r\in[n],t\in T}$ (observe that the values of $t \in T$ are not variables).

\begin{equation}
\label{eq:eq-util}
u_i= t\left(\alpha_i+\sum_{r\in N(i)} \beta_{i,r} \cdot \Fb_r(t)\right) \ \ \forall i\in [n], t\in T_i
\end{equation}
\begin{equation}
\label{eq:off-eq-util}
u_i\geq  t\left(\alpha_i+\sum_{r\in N(i)} \beta_{i,r} \cdot \Fb_r(t)\right) \ \ \forall i\in [n], t\in T \setminus T_i
\end{equation}
\begin{equation}
\label{eq:starts-0}
\Fb_i(t_k)=1\ \ \forall i\in [n]
\end{equation}
\begin{equation}
\label{eq:no-atom}
\Fb_i(1)=0\ \ \forall i\notin R_0
\end{equation}
\begin{equation}
\label{eq:yes-atom}
\Fb_i(1)>0\ \ \forall i\in R_0
\end{equation}
\begin{equation}
\label{eq:out-support}
\Fb_i(t_j)=\Fb_i(t_{j+1})\ \ \forall j\in [k-1]\ \forall\ i\notin R_j
\end{equation}
\begin{equation}
\label{eq:CDF-mon}
\Fb_i(t_j)> \Fb_i(t_{j+1})\ \ \forall j\in [k-1]\ \forall\ i\in R_j
\end{equation}
An equilibrium {\em satisfies the sketch solution} if it satisfies the sketch and, moreover, for every $i$ and $t\in T$ the value of $\Fb_i(t)$ equals the corresponding value in the sketch solution.
\end{definition}


We next explain the constraints of the linear program.
Constraints~(\ref{eq:eq-util}) state that each seller has the same utility from every boundary point in his support.
Constraints~(\ref{eq:off-eq-util}) state that each seller has weakly lower utility for boundary points that are not in his support.
Constraints~(\ref{eq:starts-0}) state that, for each $i$, the CDF for $i$ has value $0$ at the lowest boundary point.
Constraints~(\ref{eq:no-atom}) states that sellers not in $R_0$ have no atom, while
constraints~(\ref{eq:yes-atom}) state that sellers in $R_0$ have an atom.
Note that for $i\in R_0$ the size of the atom of $i$ at $1$ is exactly $\Fb_i(1)$.
Finally,
constraints~(\ref{eq:out-support}) state that sellers do not price outside their support, while
constraints~(\ref{eq:CDF-mon}) state that they do price inside their support.
Observe that as all these constraints must be satisfied in equilibrium.  If the linear program cannot be satisfied then an equilibrium satisfying the sketch does not exist.

Since the linear program can be solved in polynomial time, it follows that we can efficiently find a sketch solution for a given sketch.

\begin{observation} \label{sketch-to-ss}
There exists a polynomial time algorithm that, when given as input a nontrivial network and a sketch, it outputs a sketch solution that satisfies the sketch, if such a solution exists.
\end{observation}

We next point out that, generically, there will only be a unique sketch solution that satisfies a given sketch.

\ignore{
MOSHE: need to work on this theorem. It is not really clear what I mean by "do not violet any necessary condition for equilibrium". Maybe we want to say that if there is a unique equilibrium that satisfies the sketch then the LP will output a valid sketch solution (for that equilibrium)
\begin{observation}
\label{obs:sketch-to-solution}
For any given sketch, if there is an equilibrium that satisfies this sketch then the linear program LP1 will output a sketch solution such that
if there exists an equilibrium satisfying the sketch  it will output
values $\Fb_i(t)$ for every seller $i$ and point $t\in T$, such that these values do not violate any necessary condition for equilibrium.
Additionally, if the linear program cannot be satisfied then an equilibrium satisfying the sketch does not exist.
\end{observation}
\begin{proof}
Observe that these are indeed a set of linear equations in the variables, and that if there is no solution then no equilibrium with the specified supports and atoms exists, as each constraint represents some condition that is necessary, either for equilibrium, for the CDFs to be well defined, or for the CDFs to respect the specified supports and atoms.
\end{proof}
} 


\begin{definition}
Fix a network and a sketch. 
We say that a network has {\em full rank with respect to the sketch} if, for every $j\in [k]$, the $|R_j|\times |R_j|$ sized matrix with entries
$\beta_{i,r}$ for $(i,r)\in R_j\times R_j$ has full rank.\footnote{Note that this notion does not depend on $\vec{\alpha}$.}
\end{definition}

Note that this condition ensures that if we look at the $|R_j|$ constraints~(\ref{eq:eq-util}) as linear equations in the $|R_j|$ variables $\Fb_r(t_{j})$ for every $r\in R_j$, the system will have a unique solution.  This in turn will imply that there is at most one sketch solution that satisfies the sketch (depending on whether the unique solution to constraints~(\ref{eq:eq-util}) satisfy the remaining constraints).



We will say that a statement holds for {\em generic} values of certain parameters if the Lebesgue measure of the parameter values for which the statement holds is 1.
As any minor of a generic matrix has full rank the following observation is immediate.

\begin{observation}
\label{obs:generic-beta}
Fix a nontrivial network.  Then, for every sketch,
there exists at most one sketch solution that satisfies the sketch, 
generically over the shared market sizes $\{\beta_{i,j}\}$.
\end{observation}

\ignore{  
\begin{observation}
\label{obs:generic-alpha}
Fix a nontrivial network and a sketch.
If for  $\{\beta_{i,j}\}_{i\in [n],j\in [k]}$ the network is not $\beta$-generic with respect to the sketch, then for a generic $\alpha$
there is no equilibrium that satisfies this sketch (and no sketch solution).
\end{observation}
\begin{proof}
Since the is not $\beta$-generic with respect to the sketch, there exists $j$ such that the $|R_j|\times |R_j|$ sized matrix with entries
$\beta_{i,r}$ for $(i,r)\in R_j\times R_j$ does not have full rank. A set of $|R_j|$ linear equations in $|R_j|$ variable that does not have full rank does not have a solution when the vector of free variables is generic. Thus, there will be no solution that satisfy Constraints~(\ref{eq:eq-util}) in LP1, which are necessary for an equilibrium to exist.
\end{proof}
}


We next show that a sketch solution 
suffices for recovering the full information about a corresponding equilibrium (which will be generically unique).  Specifically, each CDF $F_i(x)$ can be described by a list of linear functions in $1/x$.  

\begin{lemma}
\label{lem:LP-to-eq}
Fix a nontrivial network and assume that we are given a sketch solution.
There is a polynomial time algorithm that outputs an equilibrium $(F_1,F_2,\ldots,F_n)$ that satisfies the sketch solution,
where each $F_i$ is a piece-wise linear functions of the inverse of its input.
Moreover, if the network has full rank with respect to the sketch, then this is the unique equilibrium that satisfies the sketch solution.
\end{lemma}
\begin{proof} 
A solution to the linear program specifies $\Fb_i(t)$ for any seller $i$ and $t\in T$. We use the solution to 
define $\Fb_j(x) = 1-F^-_j(x)$ for any seller $i$ and $x\in [0,1]$, that coincides with the solution on $\tilde{T}$.
Together with $F_i(1)=1$ for every $i$, this will completely 
define a CDF for each seller.
Once we specify the CDFs we check that they indeed form an equilibrium.


For any seller $i$ and any $j\in [k]$ we define $\Fb_i(\cdot)$ to be a linear function in $1/x$ on the interval $[t_{j+1},t_j]$, that is,
$\Fb_i(\cdot)$ is of the form $\Fb_i(\cdot)= L_{i,j}(x)=a_{i,j}+b_{i,j}/x$.
We fix the linear function to the unique linear function that coincides with the solution at the boundaries, that is
$L_{i,j}(t_{j+1})= \Fb_i(t_{j+1})$ and $L_{i,j}(t_j)=\Fb_i(t_j)$.

Given a solution to the linear program above, for seller $i$ and $t\in \tilde{T}$ define
\begin{equation}
u_i(t) = t\left(\alpha_i+\sum_{r\in N(i)} \beta_{i,r} \cdot \Fb_r(t)\right)
\end{equation}

The next lemma would be useful.
\begin{lemma}
For the CDFs as defined above, for any seller $i$, and any $j\in [k]$,
the utility $U_i(\cdot)$ is a linear function on the interval $(t_{j+1},t_j)$, moreover, it is the unique linear function that pass through the points $(t_{j+1},u_i(t_{j+1}))$ and  $(t_{j},u_i(t_{j}))$
\end{lemma}

\begin{proof}
Consider any point $x$ in the interval $(t_{j+1},t_j)$.
\begin{equation}
\label{eq:util-x}
u_i(x)= x\left(\alpha_i+\sum_{r\in N(i)} \beta_{i,r} \cdot \Fb_r(x)\right)=
\end{equation}
\begin{equation*}
x\left(\alpha_i+\sum_{r\in N(i)} \beta_{i,r}  (a_{r,j}+b_{r,j}/x)\right)=
\end{equation*}
\begin{equation*}
x\cdot \alpha_i+\sum_{r\in N(i)} \beta_{i,r} (x\cdot a_{r,j}+b_{r,j})=
\sum_{r\in N(i)} \beta_{i,r} b_{r,j} + x\left(\alpha_i+\sum_{r\in N(i)} \beta_{i,r} a_{r,j}\right)
\end{equation*}
This is clearly a linear function, and clearly it go through the two specified points by the way $\Fb_i(\cdot)$ is defined at these boundary points  for every seller $i$.
\end{proof}

This lemma shows that for the defined CDFs it is indeed the case that each $i\in R_j$ is indifferent between all the prices in the interval
$(t_{j+1},t_j)$, and that for any interval $(t_{l+1},t_l)$ such that $i\notin R_l$, $i$ cannot gain by deviating and pricing on that interval.
This prove that the specified CDFs indeed forms an equilibrium.

Finally, we observe that if the network has full rank with respect to the sketch then that equilibrium is the unique one that respects the solution to the LP. Indeed, consider any $x$ in the interval $(t_{j+1},t_j)$. The solution to LP1 specifies utility $u_i$ for every seller $i$.
For any $i\in R_j$, consider the equation $u_i=u_i(x)$ for $u_i$ as specified in Equation (\ref{eq:util-x}).
This is a  set of $|R_j|$ linear equations in the $|R_j|$ variables $\Fb_r(x)$ for every $r\in R_j$. As the network has full rank with respect to the sketch this set specifies a matrix of full rank, thus there is at must one solution to the set.
\end{proof}

The following theorem (which is a re-statement of theorem 3 from the introduction) follows by combining Lemma \ref{lem:LP-to-eq} with observations \ref{sketch-to-ss} and \ref{obs:generic-beta}.

\begin{theorem}
There is a polynomial time algorithm that gets a sketch as input and has the following properties.
If there exists an equilibrium satisfying the sketch then it will compute such an equilibrium (a list of CDFs $F_1(x),F_2(x),\ldots, F_n(x)$ each linear in $x^{-1}$),
and if such an equilibrium 
does not exist then it will provide a proof of that claim.  Moreover, generically in the shared market sizes, the provided equilibrium is unique.
\end{theorem}
\section{Trees with a Single Captive Market}
\label{sec:trees.1.captive}

We now turn our attention to a particular type of network: a tree with exactly one captive market.  Such a network may have multiple equilibria, but we will show that all equilibria have a particular form and that every equilibrium is utility-equivalent for each seller.  Moreover, when the tree is a line with the captive market at one endpoint, there is a unique equilibrium.

Fix an arbitrary tree as our network, and suppose seller $r$ is the unique seller with $\alpha_r > 0$.  We will think of the tree as being rooted at $r$.  In this rooted tree, we write $P(i)$ for the parent of seller $i$ (with $P(r) = \emptyset$), and $C(i)$ for the set of children of seller $i$.  We say $i$ is a \emph{leaf} if $C(i) = \emptyset$.  We will also write $CC(i) = \{j \colon P(P(j)) = i\}$, the set of grandchildren of $i$.

Before characterizing the equilibria of our network, it will be helpful to describe a particular type of sketch.  In this sketch, the support of each seller is an interval, say $S_v = [L_v, H_v]$.  Moreover, for each seller there is a ``midpoint'' value $M_v \in [L_v, H_v]$ such that $L_{P(v)} = M_v$ and $H_j = M_v$ for each $j \in C(v)$.  That is, the ``top'' portion of a seller's range is shared with his parent, and the ``bottom'' portion is shared with each of his children.  The root has $M_r = H_r = 1$ and each leaf $v$ has $M_v = L_v$.  Note that if sellers $v$ and $w$ are siblings then they must have $M_v = M_w$.  See Figure \ref{fig:tree}.  We say that such a profile of intervals $\{[L_v, H_v]\}_v$ is \emph{staggered}.



\begin{figure}
\centering
\begin{tabular}{c|c}
\includegraphics[scale = 0.4]{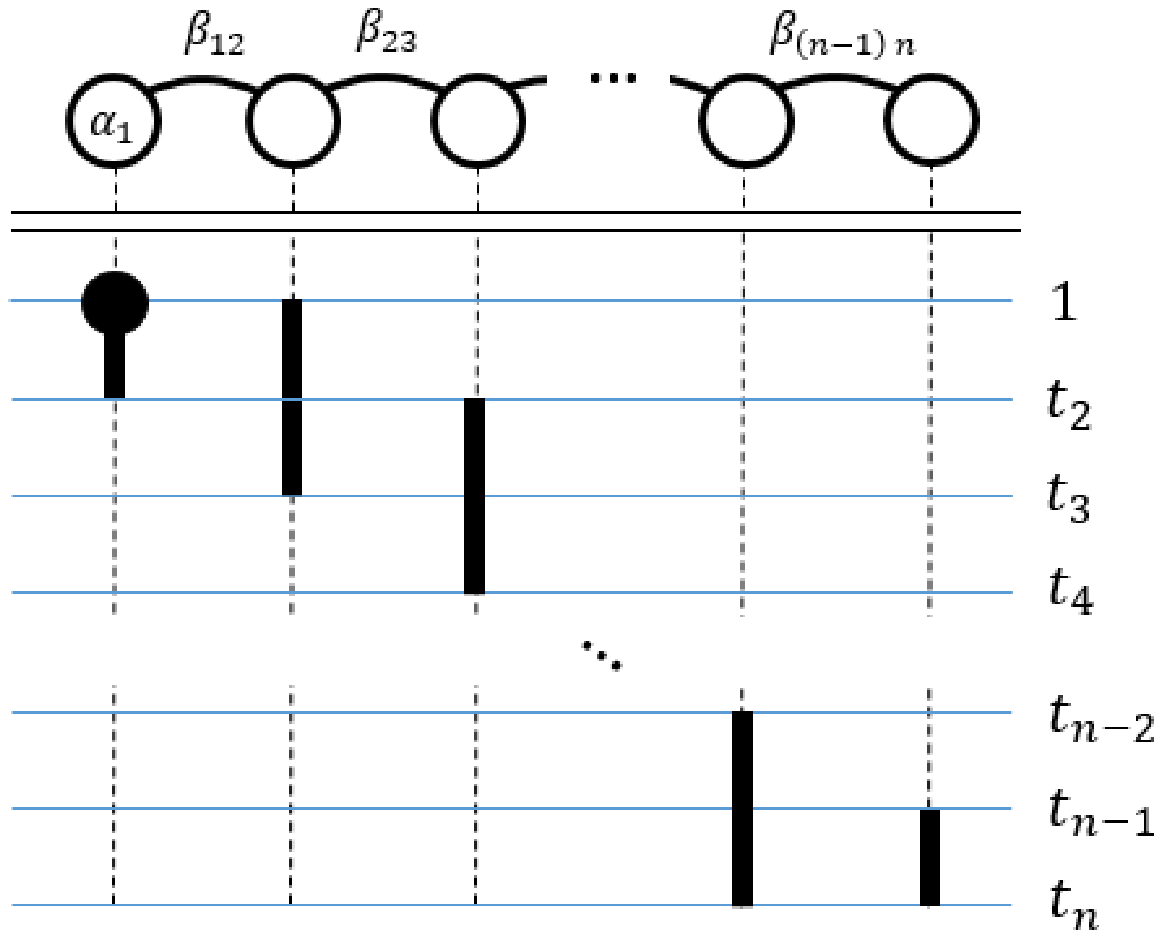} \hspace{0.5cm} &
\hspace{0.5cm} \includegraphics[scale = 0.45]{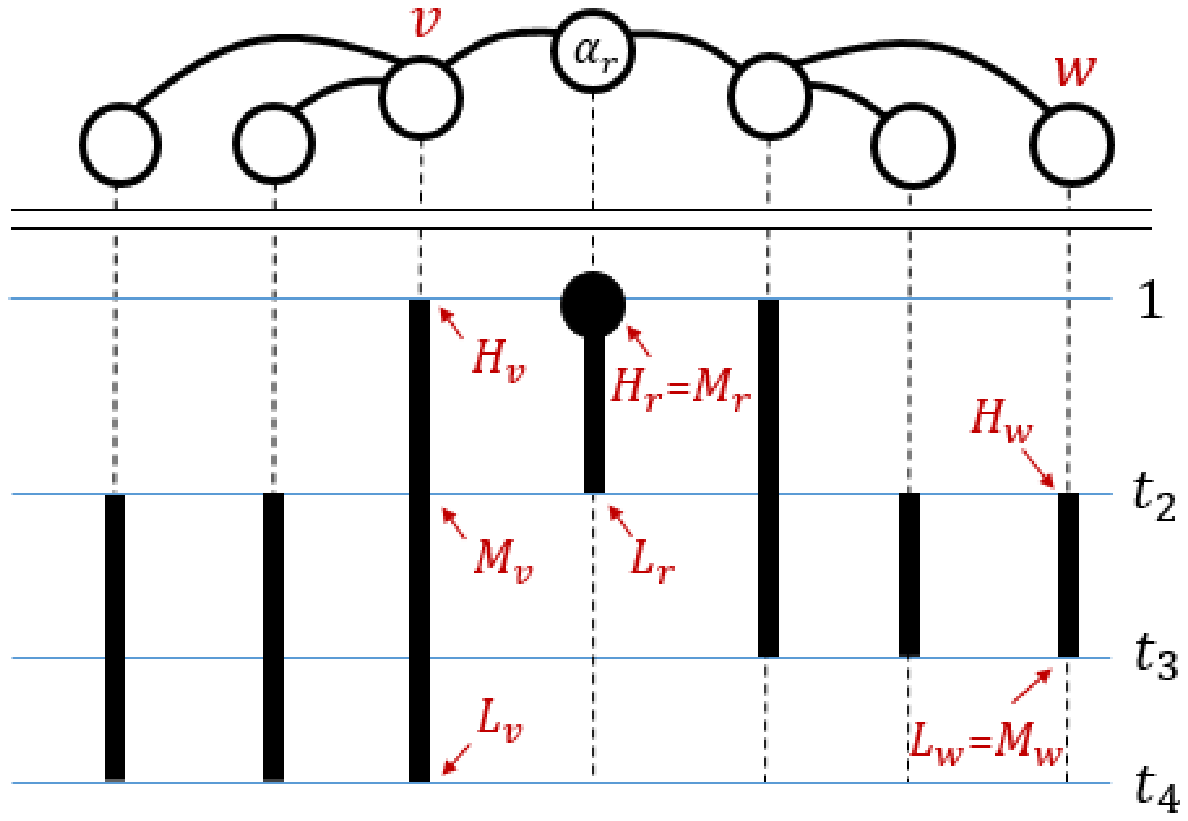} \\
(a) \hspace{0.5cm} & \hspace{0.5cm} (b)
\end{tabular}
\caption{The sketch corresponding to a staggered interval profile for (a) a line of length $n$ with a single captive market at one endpoint, and (b) a binary tree with a single captive market at its root.  Note that a profile of intervals for a tree is staggered if and only if it is staggered along each root-to-leaf path.  
In (b), the values of $L_v, M_v,$ and $H_v$ are illustrated for an interior node $v$, leaf $w$, and root $r$.
}
\label{fig:tree}
\end{figure}

Informally speaking, we will show that there is an equilibrium whose sketch corresponds to a profile of staggered intervals, where only the root has an atom at $1$.  In general this equilibrium will not be unique.  However, we will show that there is a unique profile of staggered intervals $\{[L_i, H_i]\}_i$ such that, for every equilibrium, $[L_i, M_i) \subseteq S_i \subseteq [L_i, H_i]$ for each seller $i$. Our main theorem for this section, which is a more detailed statement of Theorem \ref{ithm4} from the introduction, is as follows. 

\begin{theorem}
\label{thm:tree.single.captive}
Fix a rooted tree with single captive market, as described above.  Then there exists a profile of staggered intervals $\{[L_i, H_i]\}_i$ such that, for any equilibrium of the network and every seller $i$,
\begin{enumerate}
\item $[L_i, M_i) \subseteq S_i \subseteq [L_i, H_i]$,
\item $[L_i, M_i) \subseteq \cup_{j \in C(i)} S_j$, and
\item $\Fb_r(1) > 0$ and $\Fb_i(1) = 0$ for all $i \neq r$.
\end{enumerate}
This profile of intervals (and an equilibrium) can be computed from the network structure in polynomial time, and every equilibrium is utility-equivalent for each seller.
\end{theorem}

We prove Theorem \ref{thm:tree.single.captive} in two parts.  In Section \ref{sec:tree.equil.form} we prove that, for any given equilibrium, there is a corresponding profile of staggered intervals satisfying the conditions of Theorem \ref{thm:tree.single.captive}.  In Section \ref{sec:tree-solve} we complete the proof of Theorem \ref{thm:tree.single.captive} by showing that the profile of staggered intervals corresponding to a given equilibrium can be fully described and computed as a function of the network weights only.  This will imply that there is a single interval profile that satisfies the conditions of Theorem \ref{thm:tree.single.captive} for all equilibria.  In Section \ref{sec:tree-comp-statics} we explore some comparative statics implied by our equilibrium characterization.  In Section \ref{sec:tree-line} we focus on the special case of a line network, where we show that there is a unique equilibrium.
We will defer some proof details to Appendix \ref{app:trees.1.captive}.

\subsection{The Form of an Equilibrium}
\label{sec:tree.equil.form}

Fix a tree network as above.  In this section we prove of the following lemma.

\begin{lemma}
\label{lem:tree-staggered-intervals}
For any equilibrium $\mathbf{F} = \{F_i\}_i$, there is a profile of staggered intervals $\{[L_i^{\mathbf{F}}, H_i^{\mathbf{F}}]\}_i$ such that, for every seller $i$,
$[L_i^{\mathbf{F}}, M_i^{\mathbf{F}}) \subseteq S_i \subseteq [L_i^{\mathbf{F}}, H_i^{\mathbf{F}}]$ and
$[L_i^{\mathbf{F}}, M_i^{\mathbf{F}}) \subseteq \cup_{j \in C(i)} S_j$.
Moreover, $\Fb_r(1) > 0$ and $\Fb_i(1) = 0$ for all $i \neq r$.
\end{lemma}

Lemma \ref{lem:tree-staggered-intervals} asserts the existence of an interval profile for each equilibrium, whereas Theorem \ref{thm:tree.single.captive} makes the stronger claim that a single profile applies to \emph{all} equilibria.
%
%
Fix equilibrium $\textbf{F}$. We 
observe that the sellers' suprema must be decreasing with depth.

\begin{claim}
\label{claim.sup.desc}
We have $\Fb_r(1) > 0$ and $\Fb_1(i) = 0$ for all $i \neq r$.  Also, $\sup_i \leq \sup_{P(i)}$ for each seller $i \neq r$, with equality only if $P(i) = r$.
\end{claim}
%

We can now define a  profile of staggered intervals $\{[L_i^{\mathbf{F}}, H_i^{\mathbf{F}}]\}_i$ corresponding to equilibrium $\mathbf{F}$.  We begin with the lower bounds of the intervals.  For each $v$ with $CC(v) \neq \emptyset$, define $L_v^{\mathbf{F}} = \max_{k \in CC(v)} \sup_k$.  For $v$ with $CC(v) = \emptyset$ but $C(v) \neq \emptyset$, let $L_v^{\mathbf{F}} = \inf_v$.  Finally, for each leaf $v$, $L_v^{\mathbf{F}} = L_{P(v)}$.  For the upper bounds, set $H_r^{\mathbf{F}} = 1 = H_v^{\mathbf{F}}$ for each $v \in C(r)$, and $H_v^{\mathbf{F}} = L_{P(P(v))}$ for each $v \not\in \{r\} \cup C(r)$.  Let $M_r^{\mathbf{F}} = 1$ and $M_v^{\mathbf{F}} = L_{P(v)}$ for each $v \neq r$.  
Claim \ref{claim.sup.desc} implies that $L_v^{\mathbf{F}} \leq M_v^{\mathbf{F}} \leq H_v^{\mathbf{F}}$ for each $v$, with $L_v^{\mathbf{F}} = M_v^{\mathbf{F}}$ only if $v$ is a leaf and $M_v^{\mathbf{F}} = H_v^{\mathbf{F}}$ only if $v = r$.  Thus $\{[L_i^{\mathbf{F}}, H_i^{\mathbf{F}}]\}_i$ is, in fact, a profile of staggered intervals.

The main technical step in the proof of Lemma \ref{lem:tree-staggered-intervals} is to show that $S_v \subseteq [L_v^{\mathbf{F}}, H_v^{\mathbf{F}}]$ for each $v$.  Roughly speaking, we show that if some seller $v$ bids below $L_v^{\mathbf{F}}$, then we can find a certain pair of prices $p_1$ and $p_2$ such that $v$, a child of $v$, and a grandchild of $v$ all maximize utility at both $p_1$ and $p_2$.  We then show that such a circumstance leads to a contradiction, due to the relationship between these prices and the supports of the neighbors of the three nodes.


\begin{proposition}
\label{prop:tree.equil.Lv}
For each seller $v$, $S_v \subseteq [L_v^{\mathbf{F}}, H_v^{\mathbf{F}}]$.  Moreover, for each $j \in C(v)$ that is not a leaf, there exists $k \in C(j)$ with $\sup_k = L_v^{\mathbf{F}}$.
\end{proposition}

The other requirements of Lemma \ref{lem:tree-staggered-intervals} then follow from various applications of 
Theorem~\ref{thm:necessary-eq} (\ref{thm:nec-union-support}).

\begin{proposition}
\label{prop:children.range}
For each seller $v$, $S_v \cap [L_v^{\mathbf{F}}, M_v^{\mathbf{F}}] = [L_v^{\mathbf{F}}, M_v^{\mathbf{F}}] \cap (\cup_{j \in C(v)}S_j)$.
\end{proposition}

\begin{proposition}
\label{prop:bottom.range}
$[L_v^{\mathbf{F}}, M_v^{\mathbf{F}}) \subseteq S_v$ for each seller $v$.
\end{proposition}
%

%

Combining the results in this section completes the proof of Lemma \ref{lem:tree-staggered-intervals}.

\subsection{Uniqueness of Intervals}
\label{sec:tree-solve}

In Section \ref{sec:tree.equil.form} we defined a profile of staggered intervals $\{[L_v^{\mathbf{F}}, H_v^{\mathbf{F}}]\}_v$ for equilibrium $\mathbf{F}$.  In Appendix \ref{app:tree-solve} we show that these intervals are uniquely determined by (and can be efficiently computed from) the network, completing the proof of Theorem \ref{thm:tree.single.captive}.

\begin{lemma}
\label{lem:intervals-unique}
There exists a profile of staggered intervals $\{[L_v, H_v]\}_v$ such that for every $v$ and every equilibrium $\mathbf{F}$ it holds that 
$[L_v^{\mathbf{F}}, H_v^{\mathbf{F}}] = [L_v, H_v]$.
\end{lemma}

A corollary of our analysis is that there exists an equilibrium in which $S_v = [L_v, H_v]$ for each $v$, and in particular this equilibrium can be computed efficiently.  Another corollary is that that every equilibrium is utility-equivalent for each seller.

\begin{corollary}
All equilibria are utility-equivalent for all sellers.
\end{corollary}
\begin{proof}
Let $\{M_v\}_v$ be the profile of interval midpoints corresponding to our tree network, from Theorem \ref{thm:tree.single.captive}.  Then, for each seller $v \neq r$, we have $u_v = u_v(M_v) = M_v \beta_{v P(v)}$ in every equilibrium.  Also, $u_r = u_r(1) = \alpha_r$ in every equilibrium.  The seller utilities are therefore equilibrium-invariant, as required.
\end{proof}

\subsection{Utilities and Captive Market Size}
\label{sec:tree-comp-statics}

One implication of our equilibrium analysis is that, in every equilibrium, each seller's utility increases as $\alpha_r$ increases.

\begin{proposition}
\label{prop:tree-alpha}
For fixed shared market sizes $\boldsymbol{\beta}$, the value of $u_v$ is strictly increasing as $\alpha_r$ increases, for every seller $v$.
\end{proposition}

Our analysis in Section \ref{ss:cut} allows us to relate the size of a shared market $\beta_{v P(v)}$ to the utilities of the descendents of $v$.  Write $w \prec v$ to mean $w$ is a strict descendent of $v$.

\begin{proposition}
\label{prop:tree-beta}
Fix shared market sizes $\boldsymbol{\beta}$ and captive market size $\alpha_r$.  Choose node $v \neq r$.  If we take $\beta_{v P(v)} \to 0$, then $u_w \to 0$ for all $w \prec v$.  Alternatively, as $\beta_{v P(v)} \to \infty$, we again have $u_w \to 0$ for all $w \prec v$.
\end{proposition}

\subsection{Special Case: A Line with a Single Captive Market}
\label{sec:tree-line}

Consider now the special case that our network is a line with a single captive market belonging to one of the endpoints, $r$.  Label the sellers $i_1, \dotsc, i_n$, with $i_1 = r$ and $i_k = P(i_{k+1})$ for all $k < n$.  A corollary of Theorem \ref{thm:tree.single.captive} is that there is a \emph{unique} equilibrium.  See Figure \ref{fig:tree}(a) for an illustration of the sketch of this equilibrium.

\begin{claim}
\label{claim:line-network}
For the line network with a single captive market belonging to a seller at one endpoint, there is a unique equilibrium.  Moreover, this equilibrium has a sketch of the following form: $|T| = n$, only seller $i_1$ has an atom at $1$, and $S_{i_k} = [t_{k+1}, t_{k-1}]$ for each $k$ (where we define $t_0 = 1$ and $t_{n+1} = t_n$ for notational convenience).
\end{claim}
%
%

\begin{example}
As an illustration of our equilibrium for the line, consider the case in which $\alpha_r = 1$ and $\beta_{i_k, i_{k+1}} = 1$ for all $k < n$.  By Claim \ref{claim:line-network}, the unique equilibrium has a sketch with boundary points $T = \{t_1, \dotsc, t_n\}$, $S_{i_k} = [t_{k+1}, t_{k-1}]$ for all $1 \leq k \leq n$, and an atom at $1$ for seller $i_1$.

Considering the utility of seller $i_{n-1}$ at declarations $t_n$ and $t_{n-1}$, we have
\[ t_{n-1} = u_{i_{n-1}}(t_{n-1}) = u_{i_{n-1}}(t_{n}) = 2t_n. \]
Moreover, considering the utility of each seller $k < n-1$ at points $t_k$ and $t_{k+1}$, we have
\[ t_{k} = u_{i_{k}}(t_{k}) = u_{i_{k}}(t_{k+1}) = t_{k+1}(1 + \Fb_{i_{k+1}}(t_{k+1})) = t_{k+1} + t_{k+2} \]
where we used Claim \ref{claim.midpoints} in the last equality to infer that $t_{k+1}\Fb_{i_{k+1}}(t_{k+1}) = t_{k+2}$.  A simple recursion then implies that $t_k = N_{n-k+1} \cdot t_n$ for each $k$,
where $N_i$ denotes the $i$th Fibonacci number, indexed so that $(N_0, N_1, N_2, \dotsc) = (1, 1, 2, \dotsc)$.
Since we know $t_1 = 1$, we can solve for $t_n$ to conclude that $t_k = \frac{N_{n-k+1}}{N_n}$ for each $k$.

Since $u_{i_k} = M_{i_k} = t_k$ for each $k$, we conclude $u_{i_k} = N_{n-k+1} / N_n$ for each seller $i_k$.  In particular, the utilities of sellers decay exponentially with the distance to the captive market, with the rate of decay converging to the golden ratio as $n$ grows large.
\end{example}

\subsection{Non-uniqueness for Cycles with a Single Captive Market}
\label{sec:cycle-single-captive}
We now show that Theorem \ref{thm:tree.single.captive} does not extend to networks with cycles.
Our example is a network of $5$ sellers in a cycle, where only one seller has a captive market.  The network will be symmetric with respect to reflection about seller $3$.  We will exhibit a non-symmetric equilibrium for this network.  Symmetry will then imply the existence of two different equilibria, in which some sellers achieve different utilities.  The example is graphically depicted in figure \ref{fig:cex}(a).   The details are postponed to the appendix and give 
$u_1 = 0.645242$ and $u_5 = 0.622108$.

\begin{figure}
\centering
\begin{tabular}{c|c}
\includegraphics[scale = 0.45]{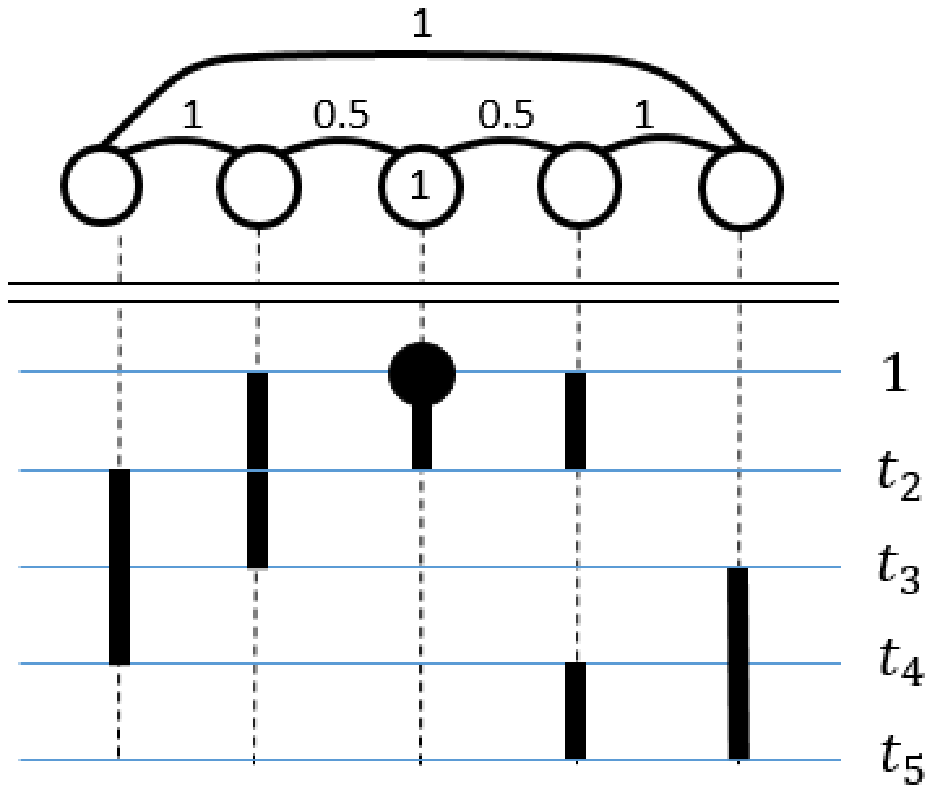} \hspace{0.5cm} &
\hspace{0.5cm} \includegraphics[scale = 0.4]{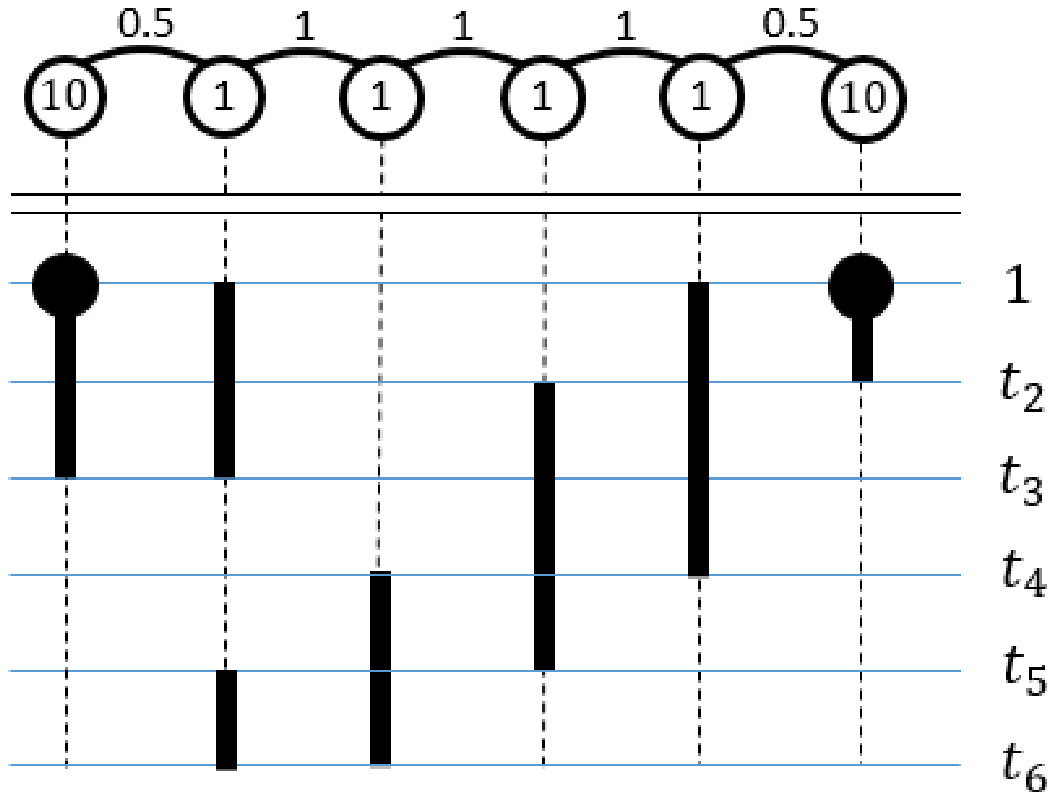} \\
(a) \hspace{0.5cm} & \hspace{0.5cm} (b)
\end{tabular}
\caption{Examples showing that Theorem \ref{thm:tree.single.captive} does not extend to (a) cycles (section \ref{sec:cycle-single-captive}) and 
(b) trees with multiple captive markets (section \ref{example:non-eq-utils}).  In each case, the network is symmetric but the equilibrium is asymmetric, 
and hence a second (reflected and non-equivalent) equilibrium exists.
}
\label{fig:cex}
\end{figure}

\section{Trees with Multiple Captive Markets}\label{sec:star}

We now consider extending our analysis of trees to allow for multiple captive markets.  As we will show, the results of Theorem \ref{thm:tree.single.captive} do not extend beyond a single captive market; we present an example with multiple equilibria that are not utility-equivalent.  However, we are able to fully characterize the (generically unique) equilibrium for the \emph{star} network with equal-sized shared markets but generic captive markets.




\subsection{Star Networks}
In this section we study star networks with all edges having the same market size, and generic sizes of the captive markets.
There is a central seller that shares a market of size $1$ with each of $n$ additional peripheral sellers.
The central seller is labelled $0$ and the $n$ sellers are labelled $1,2,\ldots,n$.
Seller $i\in \{0,1,\ldots,n\}$ has a captive market of size $\alpha_i>0$.
We assume without loss of generality that $\alpha_1\geq \alpha_2\geq \ldots \geq \alpha_n$.
All claims in this section will be made for a generic $\vec{\alpha}$.

\begin{figure}
\centering
\begin{tabular}{c|c}
\includegraphics[scale = 0.4]{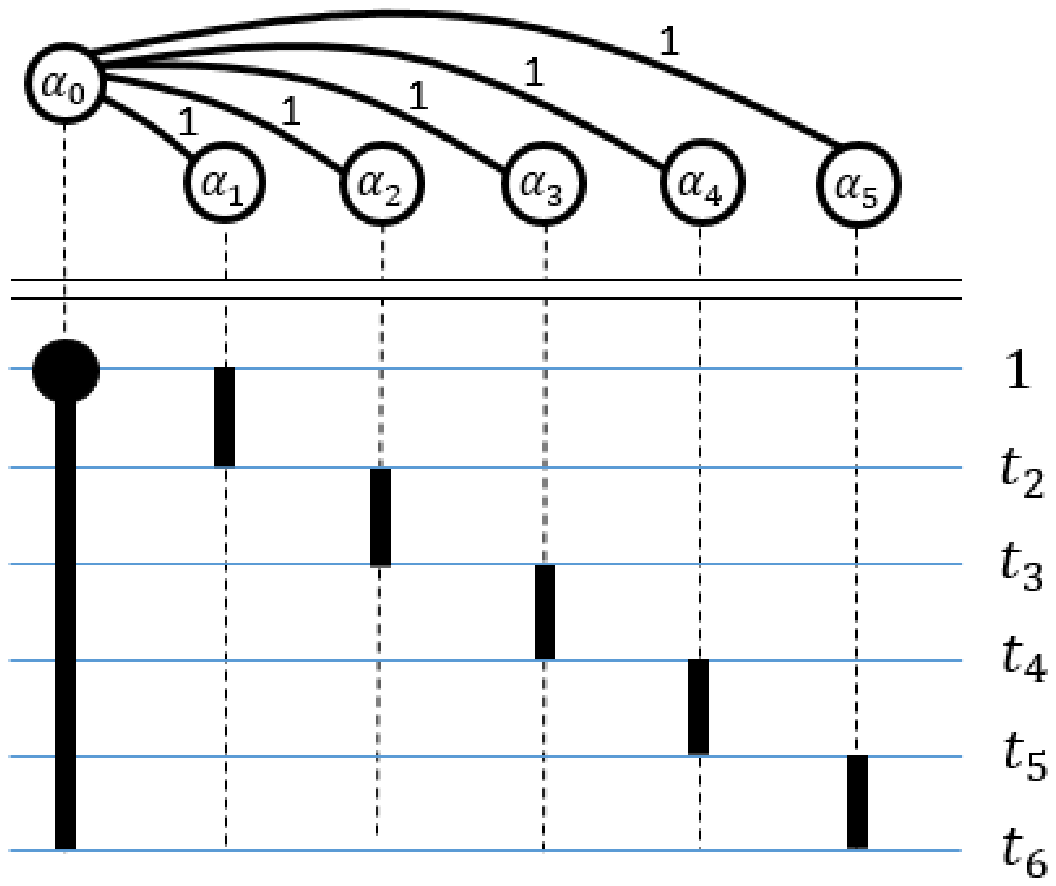} \hspace{0.5cm} &
\hspace{0.5cm} \includegraphics[scale = 0.4]{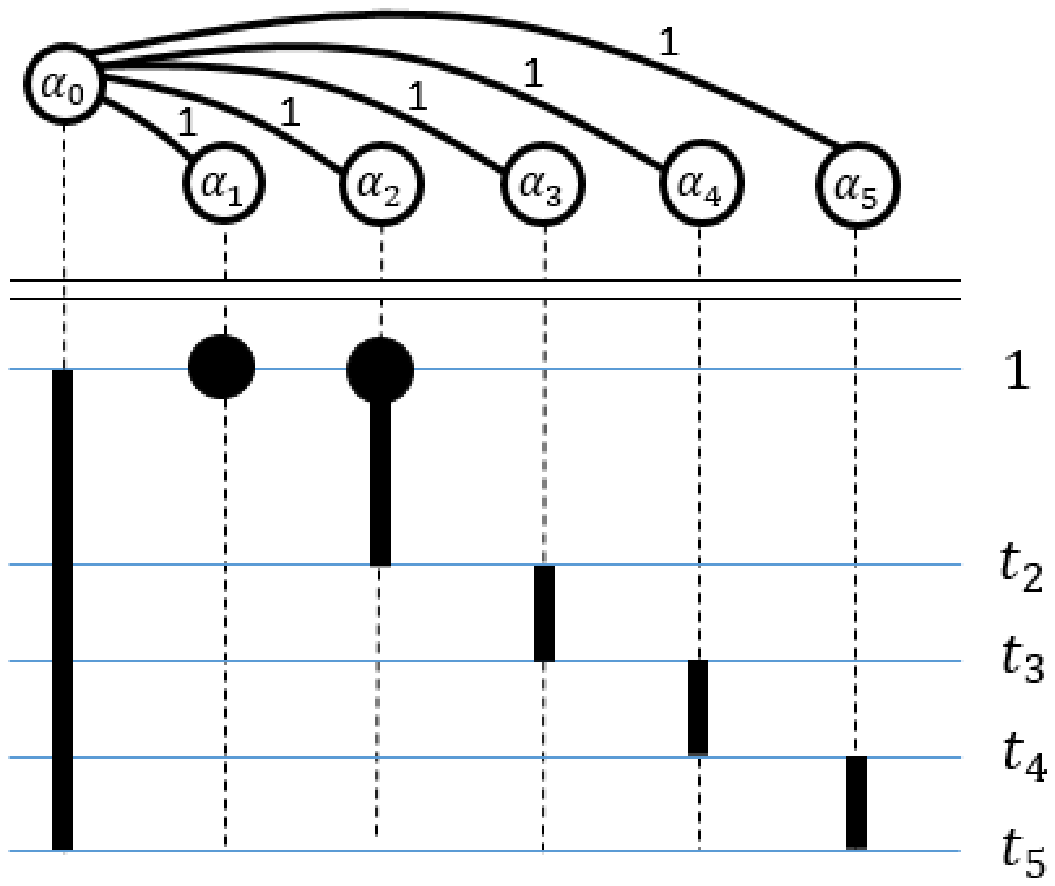} \\
(a) \hspace{0.5cm} & \hspace{0.5cm} (b)
\end{tabular}
\caption{Typical equilibria for the star network with unit shared markets and arbitrary captive markets.
}
\label{fig:star}
\end{figure}

We will show that for a star network there exists a unique equilibrium, generically with respect to $\vec{\alpha}$.
In Appendix~\ref{app:star-util} we discuss the sellers' equilibrium utilities.  
\begin{theorem}
\label{thm:star-unique}
A star network has a unique equilibrium, generically over $\vec{\alpha}$.
\end{theorem}
We outline the proof; for the complete proof see Appendix~\ref{app:star}.
We first show that for any equilibrium, a sketch that satisfies the equilibrium must have the following form.
The support of the center seller is a non-trivial interval with supremum $1$. The support of each peripheral seller is an interval (possibly degenerate, containing only point $1$).  The interiors of the peripheral sellers' intervals do not overlap.  Moreover, the intervals of the peripheral sellers are ``ordered'' by $\vec{\alpha}$: if $\alpha_i > \alpha_j$, then $S_i$ lies (weakly) above $S_j$. 
More precisely, for some $1=b_0\geq b_1\geq \ldots \geq b_n$ such that $b_n<1$, it holds that $S_0 = [b_n,1]$ and 
each peripheral seller $i$ has support $S_i=[b_i,b_{i-1}]$.
Next, 
we show that
for a star network with generic $\vec{\alpha}$,
there is a unique sketch (set of sellers with atoms at $1$, and setting of $\{b_i\}_{i\in [n]}$)
that can be satisfied in equilibrium.
For any sketch with these supports, the network has full rank with respect to the given sketch. Equilibrium uniqueness then follows from Lemma~\ref{lem:LP-to-eq}. 

\subsection{Non-uniqueness of Equilibrium: Lines with Captive Markets}
\label{example:non-eq-utils}

We now show that tree networks can exhibit multiple, non-utility-equivalent equilibria when there is more than one captive market.  Our example network will consist of $6$ sellers in a line.  This network will be symmetric, but we will exhibit a non-symmetric equilibrium.  Symmetry will then imply the existence of two different equilibria, and we will show that some sellers achieve different utilities in these equilibria.  The example is graphically depicted in
figure \ref{fig:cex}(b), the details are postponed to the appendix and show $u_5 = 1.44112$ and $u_2 = 1.4403$

\section{Quantitative Estimates of Utility}\label{util}

This section provides bounds on utility in general networks.  We study how utility "flows" from sellers with captive markets to sellers that are ``far away'' from captive markets.  We study two notions of being ``far away:'' (1) large distance in the graph to any captive market, and (2) small cut separating from all captive markets.  

\begin{theorem}
\label{thm:util}
Take a non trivial network with $n$ sellers and let $\alpha_{max}=\max_i \alpha_i$ be the size of the largest captive market.  Then 
there exist constants $c_1, c_2$ that depend only on the maximum degree of the network as 
well as on the maximum ratio between sizes
of markets in the network such that, for every seller $i$, the following are true.
\begin{enumerate}
\item
\label{thm:util-dis}
In every equilibrium, $\alpha_{max} / (c_1)^{d_i} \le u_i \le \alpha_{max} \cdot (c_1)^{d_i}$ 
where $d_i$ is the distance of $i$ from the seller with captive market $\alpha_{max}$.
\item 
\label{thm:util-cut}
Let $E$ be an edge cut that separates $i$ from all captive markets and does not contain edges adjacent to $i$, 
and change all market sizes in $E$ to be of size $\eta$ then in every equilibrium of the modified network we have that
$u_i \le {c_2}^n \alpha_{max} \cdot \eta$ and $u_i \le {c_2}^n \alpha_{max} / \eta$.
\end{enumerate}
\end{theorem}

An implication of item \ref{thm:util-cut} is that $i$'s utility goes to 0 if the cut size goes to 0 or infinity.

Subsection \ref{ss:dist} proves (a more explicit version of) part 1 of this theorem and subsection \ref{ss:cut} proves a more explicit and more general version of part
two of this theorem.  We will use the following parameters in our estimates:

\begin{enumerate}
\item The diameter of the graph denoted by $D$.
\item The ``Effective degree'' of a seller $i$: $\Delta_i=\max_{j} \frac{\alpha_i+\sum_{k \in N(i)} \beta_{ik}}{\beta_{ij}}$.  When all
market sizes ($\beta_{ij}$ and $\alpha_i$) faced by $i$ are the same
then this is exactly the degree of the seller $i$ plus $1$.  When the market sizes are not identical, this
parameter is increased by the imbalance.  The effective degree of the entire graph is $\Delta=\max_i \Delta_i$.
\item $\alpha_{max} = \max_i \alpha_i$.
\end{enumerate}

\subsection{Utilities and Distance} \label{ss:dist}


\ignore{
\begin{lemma} \label{util_captive}
In any network and any equilibrium at least one seller $i$ with a captive market $\alpha_i > 0$ gets utility $u_i = \alpha_i$.
\end{lemma}

\begin{proof}
According to proposition \ref{TBD} there exists some seller whose supremum bid is $1$, and without loss of generality we may assume
that none of his neighbors has an atom at 1 (otherwise take the neighbor, and according to lemma \ref{2} none of its neighbors will
have an atom at 1.)  When this seller bids arbitrarily close to $1$ he only wins his captive market, with total utility $\alpha_i$.
\end{proof}
} 

The following lemmas provide bounds on seller utilities with respect to the network topology.  The first bounds the possible gaps between utilities of neighbors, and the second applies this bound along a path to a captive market.

\begin{lemma} \label{nbrs}
In any network and any equilibrium, for any two sellers $j \in N(i)$, we have that $u_j \ge u_i/\Delta_i$.
\end{lemma}
\begin{proof}
Clearly $i$ can never price below $p=\frac{u_i}{\alpha_i+\sum_{k \in N(i)} \beta_{ik}}$ since even winning all his markets at a lower price would lead
to lower utility than $u_i$.  It follows that if $j$ prices at $p$ he would certainly win his whole shared market with $i$, getting utility
at least $\frac{u_i b_{ij}}{\alpha_i+\sum_{k \in N(i)} \beta_{ik}} \ge u_i/\Delta_i$ which is a lower bound to $u_j$.
\end{proof}

\begin{lemma} \label{path}
In any network and any equilibrium, for every seller $j$ we have that $\alpha_{max} / \Delta^D \le u_j \le \alpha_{max} \cdot \Delta^D$.
\end{lemma}
\begin{proof}
For the lower bound on $u_j$ take the seller $i$ with $\alpha_i=\alpha_{max}$ and apply lemma \ref{nbrs} repeatedly along the shortest path between
$i$ and $j$.  For the upper bound on $u_j$ take the seller $i$ for which Theorem~\ref{thm:necessary-eq} (\ref{thm:nec-util-alpha})
 ensures $u_i=\alpha_i \le \alpha_{max}$ and
again apply lemma \ref{nbrs} repeatedly along the shortest path between
$i$ and $j$, but this time using it to provide upper bounds.
\end{proof}

The following examples show that both a dependence in $\Delta$ and an exponential dependence in $D$ are needed.

\begin{example}
Consider a line of length $n$, with a single captive market at one end, and all markets of equal sizes.  Formally, $\alpha_1=1$, and for all $1 \le i < n$
$\beta_{i,i+1}=1$.  Thus we have $\Delta=2$ and $D=n-1$.
While $u_1=1$ per Theorem~\ref{thm:necessary-eq} (\ref{thm:nec-util-alpha}), our analysis in Section \ref{sec:tree-line} shows that $u_i = \Theta(\rho^{-i})$, where $\rho$ is the golden ratio.  Thus
we see that utilities may indeed decrease exponentially in the distance.
\end{example}

\begin{example}
Consider a line of three sellers with $\alpha_1 = 1$, $\beta_{12}=C$, and $\beta_{23}=C^2$ for some large $C$ (so in particular $\Delta=C-1$).
While $u_1=1$ per Theorem~\ref{thm:necessary-eq} (\ref{thm:nec-util-alpha}), our analysis of Example~\ref{example:line-3-single-captive} shows that $u_2 = \Theta(\Delta)$.  Thus we see three interesting
and perhaps non-intuitive effects:
first, the utility of a seller with no captive market may be larger than that of any seller with a captive market, and the gap may be unbounded.  Second,
a small captive market may increase the utility of sellers by more than its size, again with the gap being unbounded.  Third, these gaps may indeed increase
with the effective degree even when the graph size is fixed.
\end{example}

We still do not know though whether the upper bound on $u_j$ may be improved, e.g. to $u_j \le \alpha_{max} \cdot \Delta$.

\subsection{Utilities Across Cuts} \label{ss:cut}

Take a network and consider ``part'' $G$ of the network which is ``rather separated'' from 
captive markets (perhaps except from tiny ones).  We would expect that seller utilities in this part
of the market be indeed quite low.  In this section we justify this intuition for two different notions
of ``separate'', one of them quite natural and the second more surprising.

More specifically, consider an edge-cut separating $G$ from the rest of the network.  It would seem
that if the sizes of all markets on this cut are very small then the influence of the captive markets outside of $G$ cannot be too large on $G$.  This is
indeed justified by the following lemma:

\begin{lemma} \label{cut0}
Let $G$ be a subset of sellers in the network such that for every $i \in G$,
$\alpha_i + \sum_{j \not\in G} \beta_{ij} \le \epsilon$. Denote by $\Delta_G=\max_{i,j \in G} \frac{\alpha_i+\sum_{k \in N(i)} \beta_{ik}}{\beta_{ij}}$
and by $D_G$ the diameter of the largest connected component of $G$.
Then for every $i \in G$ we have that
$u_i \le \epsilon \Delta_G^{D_G}$.
\end{lemma}

\begin{proof}
We will prove it for every connected component of $G$ separately.   Take the seller with highest supremum price in a component,
and as in the proof of Theorem~\ref{thm:necessary-eq} (\ref{thm:nec-util-alpha}),
we can assume without loss of generality that none of his neighbors in $G$ has an atom at 1.  When pricing arbitrarily
close to his supremum, he will not win any markets that are shared within $G$ and thus his utility will be at most $\epsilon$.
At this point we proceed like in lemmas \ref{nbrs} and \ref{path}, only staying inside the connected component of $G$ and thus we get the same
upper bound as in lemma \ref{path} but with $\Delta$ and $D$ replaced by $\Delta_G$ and $D_G$.
\end{proof}

In particular this shows a phenomena that can be expected: take an edge-cut that separates
a subset of sellers $G$ from any captive market, and let the size of all markets in this cut approach zero, then the
utilities of all sellers in $G$ will approach zero.

The following phenomena may be more surprising: if instead of letting all market sizes in the cut approach zero, we let them approach infinity,
then it turns out that the utilities of all sellers in $G$ will also go down to zero, possibly except those who are direct neighbors of one of these cut markets that approach infinity.
The following theorem looks at a situation where a graph has two different scales of market size: ``regular'' ones and ``big'' ones, where the big
markets separate a subset of sellers $G$ from all captive markets.  We show that in this case too, the sellers in $G$ get low utility.

Consider a network where all captive markets satisfy $\alpha_i \le 1$ and all joint market sizes are either
similarly small $\beta_{ij} \le 1$ or much larger $\beta_{ij} \ge M$ (for some $M>>1$), and denote by $E$ the set of large edges
and by $B$ the sellers that share some large market.
For the following lemma we denote
$\Delta=\max(\Delta_B, \Delta_{G-B})$ and $D=\max(D_G, D_{G-B})$ where $\Delta_B$ and $D_B$ are the
maximum diameter of a connected component and the effective degree
in the sub-network that only contains the large edges within $B$ and similarly $D_{G-B}$ and $\Delta_{G-B}$ in the sub-network that only contains the small edges within $G-B$.

\begin{lemma}
If $E$ separates
a subset $G$ of the sellers from all captive markets then for all $i \in G-B$ we have that $u_i \le n^2 \Delta^{2D} / M$.
\end{lemma}

The point is that as $M$ is taken to infinity the utilities go to zero.

\begin{proof}
Let us first look at the sub-network on $B$ and try to apply lemma \ref{cut0} to it.  Notice that after scaling all market sizes down by a factor of $M$,
we have an instance where all markets that go out of $B$ have weight of at most $1/M$, and so we can apply lemma \ref{cut0} to it with $\epsilon \le n/M$
(where $n$ is the total number of sellers in the graph).  This would give us a bound of $u_i \le n \Delta_B^{D_B}/M$, where we need to emphasize that
$\Delta_B$ only takes into account the large edges in $B$ and does not depend on $M$.  (To be slightly more precise,
in case that there are
small edges between some the sellers in $B$, we need to apply a variant of lemma \ref{cut0} that defines $\Delta_G$ using only the big edges -- as we did --
but allows additional small edges as long as their total weight is also summed up as part of the markets whose weight is bounded by $\epsilon$.  The
proof of this variant is identical to the proof of lemma \ref{cut0}.  If we now scale back up all market weights by a factor of $M$,
the strategies of all sellers remain
exactly the same, but the utilities scale up by a factor of $M$ too, giving $u_i \le n \Delta_B^{D_B}$ for all $i \in B$.  This implies
that for any $i \in B$ and any possible price level $x$ we have that $1-F_j(x) \le n \Delta_B^{D_B} / (Mx)$ since otherwise the $i \in B$
that shares a large market with $j$ could price at $x$ and obtain more than $n \Delta_B^{D_B}$ utility just from this market.

Now consider a connected component of $G-B$ and take the seller $i$ in it with highest supremum price.  When he prices at his supremum he can only get
win markets that are shared with some $j \in B$ (since $E$ separates him from $G^c$). Now we can estimate his utility from above by noting
that a price of $x$ can only win the shared market $(i,j)$ with probability bounded by $1-F_j(x) \le n \Delta_B^{D_B} / (Mx)$ giving
utility of at most $n \Delta_B^{D_B} / M$ from this market an a total utility bounded by $u_i \le n^2 \Delta_B^{D_B} / M$.
For the rest of the sellers in $i$'s connected component we apply lemma \ref{path} on $G-B$ obtaining the lemma.
\end{proof}

\section{Conclusions and Open Problems}\label{sec:conclude}

We have studied price competition between sellers that have access to different sets of buyers, focusing on the case that at most two sellers can access each of the buyers' populations. Our work leaves an ample supply of open problems.
Is it possible to compute an equilibrium in any graph? We have reduced the problem of equilibrium computation to the problem of finding the supports and the set of sellers with an atom at 1, but it is not clear how to compute these in polynomial time in a general network.
Another interesting question concerns the structure of the sellers' supports in equilibrium: is the support of the equilibrium distributions necessarily finite for a generic instance?
Finally, one might like to 
consider the more general case of hyper-graphs instead of graphs, and relax our assumptions about the forms of the supply and demand curves.

\bibliographystyle{plain}
\bibliography{bib}

\appendix

\section{Trees with a Single Captive Market: Details}
\label{app:trees.1.captive}

We now provide the details of proofs from Section \ref{sec:trees.1.captive}.

\begin{claim}[Restatement of Claim \ref{claim.sup.desc}]
We have $\Fb_r(1) > 0$ and $\Fb_1(i) = 0$ for all $i \neq r$.  Also, $\sup_i \leq \sup_{P(i)}$ for each seller $i \neq r$, with equality only if $P(i) = r$.
\end{claim}
\begin{proof}
Observation \ref{obs:atoms-at-1} immediately implies $\Fb_i(1) = 0$ for all $i \neq r$.   Next, suppose for contradiction that there exists a seller $i \not\in C(r) \cup \{r\}$ such that $\sup_i \geq \sup_{P(i)}$.  Let $j$ be a seller of maximum depth such that $\sup_j \geq \sup_{P(j)}$.  The maximality of $j$ then implies $\sup_j \geq \sup_k$ for all $k \in C(j)$, and hence $\sup_j \geq \sup_\ell$ for all $\ell \in N(j)$.  This implies $u_j(\sup_j) = 0$, and hence $u_j = 0$, contradicting Observation \ref{obs:inf-price-prop}.

Next suppose $A_r(1) = 0$.  By Observation \ref{obs:contained-support}, there exists some $j \in C(r)$ with $\sup_j \geq \sup_r$.  Since we have already shown $\sup_j \geq \sup_k$ for all $k \in C(j)$, we have $u_j = 0$, contradicting Observation \ref{obs:inf-price-prop}.  Thus $\Fb_r(1) > 0$, and hence $\sup_r = 1$.  From this we can infer $\sup_j \leq \sup_{P(j)} = 1$ for all $j \in C(r)$.
\end{proof}

\begin{proposition}[Restatement of Proposition \ref{prop:tree.equil.Lv}]
For each seller $v$, $S_v \subseteq [L_v^{\mathbf{F}}, H_v^{\mathbf{F}}]$.  Moreover, for each $j \in C(v)$ that is not a leaf, there exists $k \in C(j)$ with $\sup_k = L_v^{\mathbf{F}}$.
\end{proposition}
\begin{proof}
Note that the definition of $H_v^{\mathbf{F}}$ implies that $\sup_v \leq H_v^{\mathbf{F}}$ for all $v$, so to show $S_v \subseteq [L_v^{\mathbf{F}}, H_v^{\mathbf{F}}]$ it suffices to show that $\inf_v \geq L_v^{\mathbf{F}}$.

If $v$ is a leaf, then $N(v) = \{P(v)\}$ and hence Observation \ref{obs:contained-support} implies $\inf_v \geq \inf_{P(v)} = L_v^{\mathbf{F}}$ as required.  In the case that $CC(v) = \emptyset$ but $v$ is not a leaf, we have $\inf_v = L_v^{\mathbf{F}}$ by definition.  We are left with the case that $v$ that has at least one grandchild: $CC(v) \neq \emptyset$.

Suppose for contradiction that there exists $v$ and $k \in CC(v)$ such that $\inf_v < L_v^{\mathbf{F}}$.  Let $t = \sup\{ S_v \cap [0,L_v^{\mathbf{F}}]\}$, and choose $v$ such that this value of $t$ is maximized (over all $v$ with $\inf_v < L_v^{\mathbf{F}}$), breaking ties in favor of $v$ closer to the root.

We claim that there exists $z < L_v^{\mathbf{F}}$ such that $z \in S_v$ and either $v = r$ or $F_{P(v)}(z) = F_{P(v)}(L_v^{\mathbf{F}})$.  If $v = r$ then this follows immediately from the fact that $\inf_v < L_v^{\mathbf{F}}$.  Otherwise, if $F_{P(v)}(z) < F_{P(v)}(L_v^{\mathbf{F}})$ for all $z \in S_v \cap [0, L_v^{\mathbf{F}})$, then we conclude that $\sup\{ S_{P(v)} \cap [0,L_{P(v)}^{\mathbf{F}}]\} \geq \sup\{ S_{P(v)} \cap [0,L_v^{\mathbf{F}}]\} \geq \sup\{ S_v \cap [0,L_v^{\mathbf{F}}]\}$, contradicting our choice of $v$.  We therefore have $F_{P(v)}(z) = F_{P(v)}(L_v^{\mathbf{F}})$ as claimed.

If $v \neq r$, we have
\[ u_v(z) = z ( \Fb_{P(v)}(z)\beta_{v P(v)} + \sum_{j \in C(v)}\Fb_j(z)\beta_{v j} ) \]
and
\[ u_v(L_v^{\mathbf{F}}) = L_v^{\mathbf{F}} ( \Fb_{P(v)}(L_v^{\mathbf{F}})\beta_{v P(v)} + \sum_{j \in C(v)}\Fb_j(L_v^{\mathbf{F}})\beta_{v j} ). \]
Since $F_{P(v)}(L_v^{\mathbf{F}}) = F_{P(v)}(z) < 1$, we can write $X = \Fb_{P(v)}(z)\beta_{v P(v)}$ and conclude
\begin{equation}
\label{eq.zLv}
zX + \sum_{j \in C(v)}z\Fb_j(z)\beta_{v j} = u_v(z) \geq u_v(L_v^{\mathbf{F}}) = L_v^{\mathbf{F}} X + \sum_{j \in C(v)}L_v^{\mathbf{F}} \Fb_j(L_v^{\mathbf{F}})\beta_{v j}
\end{equation}
where $X > 0$.  If $v = r$ we also have \eqref{eq.zLv} with $X = \Fb_r(1) > 0$, so \eqref{eq.zLv} holds for any choice of $v$.

Our strategy for the remainder of the proof of Proposition \ref{prop:tree.equil.Lv} is to show that $L_v^{\mathbf{F}}\Fb_j(L_v^{\mathbf{F}}) \geq z\Fb_j(z)$ for each $j \in C(v)$.  This plus the fact that $zX < L_v^{\mathbf{F}} X$ will contradict \eqref{eq.zLv}, leading to the desired contradiction.

From the definition of $L_v$, there must exist $j \in C(v)$ and $k \in C(j)$ with $\sup_k = L_v^{\mathbf{F}}$, and in particular $1 = F_k(L_v^{\mathbf{F}}) > F_k(w)$ for all $w < L_v^{\mathbf{F}}$.  Since $L_v^{\mathbf{F}} \in S_j$, we then have $u_j(L_v^{\mathbf{F}}) \geq u_j(w)$ for all $w \leq L_v^{\mathbf{F}}$, where
\[ u_j(L_v^{\mathbf{F}}) = L_v^{\mathbf{F}} \beta_{v j}\Fb_v(L_v^{\mathbf{F}})\]
and
\[ u_j(w) \geq w( \beta_{v j}\Fb_v(w) + \beta_{j k}\Fb_k(w) ) > w \beta_{v j}\Fb_v(w) \]
so, in particular, $L_v^{\mathbf{F}}\Fb_v(L_v^{\mathbf{F}}) > w\Fb_v(w)$ for all $w < L_v^{\mathbf{F}}$.

Choose $j \in C(v)$ and first suppose that $j$ is a leaf.  Then $u_j(L_v^{\mathbf{F}}) = L_v^{\mathbf{F}}\Fb_i(L_v^{\mathbf{F}}) > w\Fb_v(w) = u_j(w)$ for all $w < L_v^{\mathbf{F}}$.  We conclude that $[0, L_v^{\mathbf{F}}) \cap S_j = \emptyset$, and in particular $F_j(L_v^{\mathbf{F}}) = F_j(z)$, so $L_v^{\mathbf{F}}\Fb_j(L_v^{\mathbf{F}}) \geq z\Fb_j(z)$ as required.

Suppose $j$ is not a leaf.  From the definition of $L_v^{\mathbf{F}}$, $\sup_k \leq L_v^{\mathbf{F}}$ for all $k \in C(j)$.  We claim that there exists $k \in C(j)$ with $\sup_k = L_v^{\mathbf{F}}$ (this will establish the second claim of Proposition \ref{prop:tree.equil.Lv}).  Let $x = \max_{k \in C(j)} \sup_k$, and suppose for contradiction that $x < L_v^{\mathbf{F}}$.  We must then have $x \in S_j$, since Claim \ref{claim.sup.desc} implies that $j$ is the only neighbor of $k$ who can price at $\sup_k$, for each $k \in C(j)$.  But then $u_j(L_v^{\mathbf{F}}) = L_v^{\mathbf{F}}\Fb_v(L_v^{\mathbf{F}}) > x\Fb_v(x) = u_j(x)$, a contradiction.

We conclude that there exists $k \in C(j)$ with $\sup_k = L_v^{\mathbf{F}}$, and hence $u_k(L_v^{\mathbf{F}}) \geq u_k(z)$ for this $k$.  Since all children of $k$ have suprema strictly less than $\sup_k = L_v^{\mathbf{F}}$ by Claim \ref{claim.sup.desc}, we then have
\[ L_v^{\mathbf{F}}\Fb_j(L_v^{\mathbf{F}}) \beta_{j k} = u_k(L_v^{\mathbf{F}}) \geq u_k(z) \geq z\Fb_j(z) \beta_{j k} \]
and hence $L_v^{\mathbf{F}}\Fb_j(L_v^{\mathbf{F}}) \geq z\Fb_j(z)$ as required.
\end{proof}

\begin{proposition}[Restatement of Proposition \ref{prop:children.range}]
For each seller $v$, $S_v \cap [L_v^{\mathbf{F}}, M_v^{\mathbf{F}}] = [L_v^{\mathbf{F}}, M_v^{\mathbf{F}}] \cap (\cup_{j \in C(v)}S_j)$.
\end{proposition}
\begin{proof}
If $v$ is a leaf then this is vacuously true.  Note that $(L_v^{\mathbf{F}}, M_v^{\mathbf{F}}) \cap S_{k} = \emptyset$ for each $k \in CC(v)$, since $L_v^{\mathbf{F}} = H_k^{\mathbf{F}} \geq \sup_k$ for any such $k$ by Proposition \ref{prop:tree.equil.Lv}.  Thus, for any $j \in C(v)$, we must have $S_j \cap [L_v^{\mathbf{F}}, M_v^{\mathbf{F}}] \subseteq S_v$ by Observation \ref{obs:contained-support}.  Similarly, if $v \neq r$ then $\inf_{P(v)} \geq L_{P(v)}^{\mathbf{F}} = M_v^{\mathbf{F}}$, $(L_v^{\mathbf{F}}, M_v^{\mathbf{F}}) \cap S_{P(v)} = \emptyset$.  We must therefore have $S_v \cap [L_v^{\mathbf{F}}, M_v^{\mathbf{F}}] \subseteq \cup_{j \in C(v)} S_j$, again by Observation \ref{obs:contained-support}.
\end{proof}

\begin{proposition}[Restatement of Proposition \ref{prop:bottom.range}]
$[L_v^{\mathbf{F}}, M_v^{\mathbf{F}}) \subseteq S_v$ for each seller $v$.
\end{proposition}
\begin{proof}
We first claim that $L_v^{\mathbf{F}} \in S_v$.  If $CC(v) = \emptyset$ then $L_v^{\mathbf{F}} = \inf_v \in S_v$ by definition.  Otherwise, there exists some $j \in C(v)$ and $k \in C(j)$ with $\sup_k = L_v^{\mathbf{F}}$ by Proposition \ref{prop:tree.equil.Lv}.
We must therefore have $L_v^{\mathbf{F}} \in S_j$ by Observation \ref{obs:contained-support}, and hence $L_v^{\mathbf{F}} \in S_v$ by Proposition \ref{prop:children.range}.

Let $x = \sup\{x \colon [L_v^{\mathbf{F}}, x] \subseteq S_v\}$.  Since $L_v^{\mathbf{F}} \in S_v$ we know $x \geq L_v^{\mathbf{F}}$.  Suppose for contradiction that $x < M_v^{\mathbf{F}}$.  Then there exists some range $(x, x+\epsilon) \subseteq [L_v^{\mathbf{F}}, M_v^{\mathbf{F}})$ such that $x \in S_v$ but $(x, x+\epsilon) \cap S_j = \emptyset$ for each $j \in C(v)$.  However, since $\inf_{P(v)} \geq L_{P(v)}^{\mathbf{F}} = M_v^{\mathbf{F}}$, we then conclude that
\[ u_v(x+\epsilon) = (x+\epsilon)(\beta_{v P(v)} + \sum_{j \in C(v)}\beta_{j v}\Fb_j(x)) > x(\beta_{v P(v)} + \sum_{j \in C(v)}\beta_{j v}\Fb_j(x))  = u_v(x),\]
a contradiction.  Thus $x = M_v^{\mathbf{F}}$, so $[L_v^{\mathbf{F}}, M_v^{\mathbf{F}}) \subseteq S_v$ as required.
\end{proof}

\subsection{Uniqueness of Intervals}
\label{app:tree-solve}

In this section we prove Lemma \ref{lem:intervals-unique}, which is that
there exists a profile of staggered intervals $\{[L_v, H_v]\}_v$ such that $\{[L_v^{\mathbf{F}}, H_v^{\mathbf{F}}]\}_v = \{[L_v, H_v\}_v$ for every equilibrium $\mathbf{F}$.

Choose an equilibrium $\mathbf{F}$ and consider the corresponding profile of staggered intervals $\{[L_v^{\mathbf{F}}, H_v^{\mathbf{F}}\}_v$ (and values $M_v^{\mathbf{F}}$).  Write $w \succ v$ to mean that $w$ is a (strict) ancestor of $v$.  We now show that each $M_v^{\mathbf{F}}$ is determined by the values of $\Fb_w(M_w^{\mathbf{F}})$ for the ancestors of $v$.

\begin{claim}
\label{claim.midpoints}
$M_v^{\mathbf{F}} = \prod_{w \succ v}\Fb_{w}(M_w^{\mathbf{F}})$
\end{claim}
\begin{proof}
If $v = r$ then $M_r^{\mathbf{F}} = 1$ as required.  Otherwise, let $w = P(v)$.  By Proposition \ref{prop:children.range} there exists some $j \in C(w)$ with $M_w^{\mathbf{F}} \in S_j$.  For this seller $j$, we have $u_j(M_j^{\mathbf{F}}) = u_j(M_w^{\mathbf{F}})$, which implies
\[ M_j^{\mathbf{F}} \beta_{j w} = M_w^{\mathbf{F}} \beta_{j w} \Fb_w( M_w^{\mathbf{F}} ). \]
Since $M_j^{\mathbf{F}} = M_v^{\mathbf{F}}$ from the definition of a staggered interval profile, we then have $M_v^{\mathbf{F}} = M_w^{\mathbf{F}} \Fb_w( M_w^{\mathbf{F}} )$.  The result then follows by structural induction on the tree, with base case $M_r^{\mathbf{F}} = 1$.
\end{proof}

\begin{claim}
\label{claim.cdf.midpoints}
For each seller $v$, $\Fb_v(M_v^{\mathbf{F}})$ is positive and uniquely determined by the network weights.  Moreover, $\Fb_v(M_v^{\mathbf{F}})$ is independent of $\alpha_r$ for each $v \neq r$.
\end{claim}
\begin{proof}
We proceed by structural induction on the tree network.  For the base case, suppose $v$ is a leaf; then $\Fb_v(M_v^{\mathbf{F}}) = \Fb_v(L_v^{\mathbf{F}}) = 1$.  Next suppose that $C(v) \neq \emptyset$ and $v \neq r$.  By Lemma \ref{lem:tree-staggered-intervals} we have that $[M_v^{\mathbf{F}}, L_v^{\mathbf{F}}) \subseteq S_v$ and hence $u_v(M_v^{\mathbf{F}}) = u_v = u_v(L_v^{\mathbf{F}})$.  Choose $j \in C(v)$ and note that $L_v^{\mathbf{F}} = M_j^{\mathbf{F}} = M_v^{\mathbf{F}} \Fb_v(M_v^\mathbf{F})$ by Claim \ref{claim.midpoints}.  We conclude that
\[ M_v^{\mathbf{F}} \beta_{v P(v)} = M_v^{\mathbf{F}} \Fb_v(M_v^{\mathbf{F}}) \left( \beta_{v P(v)} + \sum_{j \in C(v)}\beta_{j v}\Fb_j(M_j^{\mathbf{F}}) \right) \]
which implies
\begin{equation}
\label{eq.Fb.v}
\Fb_v(M_v^{\mathbf{F}}) = \frac{\beta_{v P(v)}}{\beta_{v P(v)} + \sum_{j \in C(v)}\beta_{j v}\Fb_j(M_j^{\mathbf{F}})} > 0.
\end{equation}
Since each $\Fb_j(M_j^{\mathbf{F}})$ is uniquely determined by the network weights by induction, we conclude that $\Fb_v(M_v^{\mathbf{F}})$ is as well.
Finally, for $v = r$, a similar analysis yields
\begin{equation}
\label{eq.Fb.root}
\Fb_r(M_r^{\mathbf{F}}) = \frac{\alpha_r}{\alpha_r + \sum_{j \in C(r)}\beta_{j r}\Fb_j(M_j^{\mathbf{F}})} > 0
\end{equation}
and hence our induction implies $\Fb_r(M_r^{\mathbf{F}}) = \Fb_r(1)$ is positive and uniquely determined by the network weights, as required.
\end{proof}

Claims \ref{claim.midpoints} and \ref{claim.cdf.midpoints} together imply that the values $M_v^{\mathbf{F}}$ are uniquely determined by the network weights.  Note that this specifies the values of $L_v^{\mathbf{F}}$ and $H_v^{\mathbf{F}}$ as well, since $H_v^{\mathbf{F}} = M_{P(v)}^{\mathbf{F}}$ and $L_v^{\mathbf{F}} = M_j^{\mathbf{F}}$ for any $j \in C(v)$ (or $L_v^{\mathbf{F}} = M_v^{\mathbf{F}}$ if $C(v) = \emptyset$).

We conclude that interval profile $\{[L_v^{\mathbf{F}}, H_v^{\mathbf{F}} ]\}_v$ is independent of $\mathbf{F}$, and hence there is a unique profile $\{[L_v, H_v ]\}_v$ corresponding to \emph{every} equilibrium.  This completes the proof of Lemma \ref{lem:intervals-unique}, and of Theorem \ref{thm:tree.single.captive}.

\subsection{Utilities and Captive Market Size}

In this section we prove Proposition \ref{prop:tree-alpha}, which states that for fixed edge weights $\boldsymbol{\beta}$, the value of $u_v$ is strictly increasing as $\alpha_r$ increases, for every seller $v$.

\begin{observation}
\label{obs.2}
For every seller $v \neq r$, there exist positive $X_v, Y_v, Z_v$ independent of $\alpha_r$ such that
\[ M_v = \frac{\alpha_r X_v}{\alpha_r Y_v + Z_v} \]
\end{observation}
\begin{proof}
We proceed by structural induction on the tree network, top-down.   Suppose first that $v \in C(r)$.  Claim \ref{claim.midpoints} and \eqref{eq.Fb.root} imply that
\[ M_v = \Fb_r(1) = \frac{\alpha_r}{\alpha_r+\sum_{j \in C(r)} \beta_{j r} \Fb_j(M_j)}. \]
Since each $\Fb_j(M_j)$ is positive and independent of $\alpha_r$ (from Claim \ref{claim.cdf.midpoints}), we have that $M_v$ is of the required form.

For $v \not\in C(r) \cup \{r\}$, Claim \ref{claim.midpoints} and \eqref{eq.Fb.v} imply
\[ M_v = M_{P(v)} \Fb_{P(v)}(M_{P(v)}) = M_{P(v)} \frac{\beta_{P(v) P(P(v))} }{\beta_{P(v) P(P(v))} + \sum_{j \in C(P(v))} \beta_{j P(v)} \Fb_j(M_j)}. \]
Since each $\Fb_j(M_j)$ is positive and independent of $\alpha_r$, we have by induction that
\[ M_v = \frac{\alpha_r X_{P(v)}}{\alpha_r Y_{P(v)} + Z_{P(v)}} \cdot \frac{W}{U} \]
where $W$ and $U$ are positive and do not depend on $\alpha_r$, and hence $M_v$ is of the required form.
\end{proof}

We can now complete the proof of Proposition \ref{prop:tree-alpha}.
%
If $v = r$ then we have $u_r = \alpha_r$ so the result holds.  For $v \neq r$, we have
\[ u_v = M_v \beta_{v P(v)} = \frac{\alpha_r X_v}{\alpha_r Y_v + Z_v} \]
where each of $X_v, Y_v, Z_v$ is positive and independent of $\alpha_r$, by Observation \ref{obs.2}.  This implies that $u_v$ is strictly increasing in $\alpha_r$.

\subsection{Special Case: A Line with a Single Captive Market}

\begin{claim}[Restatement of Claim \ref{claim:line-network}]
For the line network with a single captive market belonging to a seller at one endpoint, there is a unique equilibrium.  Moreover, this equilibrium has a sketch of the following form: $|T| = n$, only seller $i_1$ has an atom at $1$, and $S_{i_k} = [t_{k+1}, t_{k-1}]$ for each $k$ (where we define $t_0 = 1$ and $t_{n+1} = t_n$ for notational convenience).
\end{claim}
\begin{proof}
By Theorem \ref{thm:tree.single.captive}, there is a staggered profile of intervals such that $[L_i, M_i) \subseteq S_i \subseteq [L_i, H_i]$, and moreover $[L_i, M_i) \subseteq \cup_{j \in C(i)} S_j$.  Since $C(i_k) = i_{k+1}$ for all $k < n$, we have $\cup_{j \in C(i_k)} S_j = S_{i_{k+1}}$, and hence $[L_{i_k}, M_{i_k}) \subseteq S_{i_{k+1}}$ for each $k < n$.  Since $L_{i_k} = M_{i_{k+1}}$ and $M_{i_k} = H_{i_{k+1}}$, we conclude that $[M_i, H_i) \subseteq S_i$ for each $i \neq r$.  We therefore have $S_i = [L_i, H_i]$ for each seller $i$.

We can now describe the sketch of our equilibrium.  We have a set of boundary points $T = \{t_1, \dotsc, t_n\}$ with $t_j = M_{i_j}$ for each $j \leq n$.  Our supports are of the form $S_{i_k} = [L_{i_k}, H_{i_k}] = [t_{k-1}, t_{k+1}]$ for each $i$.

The analysis of Section \ref{sec:tree-solve} provides a sketch solution for this sketch.  Noting that the line network has full rank over this sketch, Lemma \ref{lem:LP-to-eq} implies that there is a unique equilibrium satisfying this sketch solution, and hence a unique equilibrium for our network.
\end{proof}

\subsection{Non-uniqueness for Cycles with a Single Captive Market}

Here are the details of the example depicted in figure \ref{fig:cex}(a).
The set of sellers is $\{1, 2, 3, 4, 5\}$.  We will have $\alpha_3 = 1$ and $\alpha_j = 0$ for $j \neq 3$.  The shared market weights are $(\beta_{12}, \beta_{23}, \beta_{34}, \beta_{45}, \beta_{51}) = (1,0.5,0.5,1,1)$.  Note that this market is symmetric in terms of reflection around seller $3$.

We now describe a sketch for this network.
Our set of boundary points will be $T = (t_1, t_2, t_3, t_4, t_5)$.  Seller $3$ is the only one with an atom at $1$.
The sellers' supports are $S_1 = [t_4, t_2]$, $S_2 = [t_3, t_1]$, $S_3 = [t_2, t_1]$, $S_4 = [t_5, t_4] \cup [t_2, t_1]$, and $S_5 = [t_5, t_3]$.

Thinking of program (LP1) from Section \ref{sec:sketches} as a quadratic program in which the boundary points $t_i$ are treated as variables, we can solve to find a sketch solution.  The boundary points of this solution are (approximately)
\[ (t_1, t_2, t_3, t_4, t_5) = (1, 0.933163, 0.645242, 0.357321, 0.311054) \]
and the relevant values of $\Fb_i(t_j)$ are given by
\[ \Fb_1(t_3) = 0.223111 \quad\quad \Fb_2(t_2) = 0.691457 \quad\quad \Fb_3(1) = 0.933163\]
\[ \Fb_4(t_4) = 0.741037 \quad\quad \Fb_5(t_4) = 0.805778 \]
At this equilibrium, $u_5 = t_5( \alpha_5 + \beta_{4 5} + \beta_{5 1} ) = 0.622108$ and $u_1 = t_3( \alpha_1 + \beta_{1 2} ) = 0.645242$.

By symmetry of the network, there exists a second equilibrium in which the cycle is reflected about seller $3$; that is, with the roles of sellers $2$ and $4$ reversed, and the roles of sellers $1$ and $5$ reversed.
In this equilibrium, $u_5 = 0.645242$ and $u_1 = 0.622108$.  These two equilibria are therefore not utility-equivalent for the sellers, as required.

\section{Star: Proof of Theorem~\ref{thm:star-unique}}
\label{app:star}
We restate and prove Theorem~\ref{thm:star-unique}.

\begin{theorem}
For a star network with generic $\vec{\alpha}$, there exists a unique equilibrium.
\end{theorem}
We first outline the proof.
To prove the claim we first show that for any equilibrium, a sketch that satisfies the equilibrium must has the following form.
The center price on a non-trivial interval with supremum $1$, and the peripheral sellers each price on an interval (possibly degenerated to the point $1$), the interior of these intervals do not overlap. Moreover, the interval of a peripheral seller is above the intervals of any other peripheral seller with a smaller captive market.
Formally, for some $1=b_0\geq b_1\geq \ldots \geq b_n$ such that $b_n<1$ it holds that the support of the center is the interval $[b_n,1]$.
Additionally, Each peripheral seller $i$ has support $S_i=[b_i,b_{i-1}]$.
Next, in Lemma~\ref{lem:star-sketch} we show that
for a star network with generic $\vec{\alpha}$,
there is a unique sketch (set of sellers with atoms at $1$, and setting of $\{b_i\}_{i\in [n]}$)
that can be satisfied in equilibrium.
For any sketch with these supports, the network has full rank with respect to the given sketch. Equilibrium uniqueness follows from Lemma~\ref{lem:LP-to-eq}. 

Consider any equilibrium in this market.
For any seller $i$ and any point $x\in [0,1]$, recall that $\Fb_i(x) = 1-F^-_i(x)$.
As each peripheral node has only one neighbor (the center), and its support (except possibly an atom at 1) must be contained in the center's support (Observation~\ref{obs:contained-support}), the center cannot be pricing at $1$ with probability $1$.
Additionally the same observation implies that there is at least one peripheral seller that is not always pricing at $1$.

\begin{observation}
\label{obs:closure-intersection}
For a  star network with $\alpha_1> \alpha_2> \ldots > \alpha_n>0$
in any equilibrium the intersection of the supports of any two peripheral sellers includes at most one point.
\end{observation}
\begin{proof}
Consider two peripheral sellers $i,j$ and assume that both $x'$ and $x''\neq x'$ are in the support of both sellers, and optimal for them.
That is,
$$(\alpha_i+\Fb_0(x'))x'= (\alpha_i+\Fb_0(x''))x''$$
and
$$(\alpha_j+\Fb_0(x'))x'= (\alpha_j+\Fb_0(x''))x''$$
Thus
$$\alpha_i = \frac{\Fb_0(x'')x'' - \Fb_0(x')x'}{x'-x''} = \alpha_j$$
a contradiction to $\alpha_i\neq \alpha_j$ for every $i\neq j$.
\end{proof}

\begin{observation}
\label{obs:no-gap} For a star network,  in any equilibrium the support of the center is an interval with supremum of $1$, and this interval is exactly the union of the supports of the peripheral sellers.
\end{observation}

\begin{observation}
\label{obs:monotone-utilities}
For a star network, in any equilibrium, for any peripheral sellers $i<j$ with $\alpha_i\geq \alpha_j$ it holds that $u_i\geq u_j$, with strict inequality if $\alpha_i> \alpha_j$.
\end{observation}
\begin{proof}
For any $x\in S_j$
$$u_j = u_j(x)= (\alpha_j+\Fb_0(x))x$$

The utility of $i$ is at least his utility by pricing at some $x\in S_j$ thus
$$u_i \geq u_i(x)= (\alpha_i+\Fb_0(x))x\geq (\alpha_j+\Fb_0(x))x= u_j(x)=u_j$$
when the right inequality follows since $\alpha_i\geq \alpha_j$ and is strict if $\alpha_i> \alpha_j$ (since $x\in S_j$ means $x>0$, by Corollary~\ref{cor:captive-positive}).
\end{proof}

\begin{observation}
\label{obs:no-crossing}
Fix any star network and any equilibrium. For any pair of peripheral sellers  $i<j$ with $\alpha_i>\alpha_j$ it holds that any price in the support of $i$ is at least as high as any price in the support of $j$.
That is for any $x'\in S_i$ and $x''\in S_j$ it holds that $x'\geq x''$.
\end{observation}
\begin{proof}

Assume in contradiction that $x'< x''$ for $x'\in S_i$ and $x''\in S_j$.
We will show that seller $j$ can increase his utility by pricing at $x'$ instead of $x''$.
As $x'\in S_i$ it holds that $$u_i=u_i(x')= (\alpha_i+\Fb_0(x'))x'\geq  (\alpha_i+\Fb_0(x''))x'' = u_i(x'')$$
Thus,
$$x' \geq  x'' \frac{\alpha_i+\Fb_0(x'')}{\alpha_i+\Fb_0(x')}$$

Combining with $x''$ being optimal for $j$ (as $x''\in S_j$), it holds that
$$u_j=u_j(x'')= (\alpha_j+\Fb_0(x''))x''\geq u_j(x') = x'(\alpha_j+\Fb_0(x'))\geq x'' \frac{\alpha_i+\Fb_0(x'')}{\alpha_i+\Fb_0(x')}(\alpha_j+\Fb_0(x'))$$
we conclude that
$$
(\alpha_j+\Fb_0(x''))\geq \frac{\alpha_i+\Fb_0(x'')}{\alpha_i+\Fb_0(x')}(\alpha_j+\Fb_0(x'))
$$
Simplifying this shows that this is equivalent to $\alpha_i\leq \alpha_j$, a contradiction.
\end{proof}

\begin{corollary}
For a  star network with $\alpha_1> \alpha_2> \ldots > \alpha_n>0$,
any equilibrium must have the following form.
For some $1=b_0\geq b_1\geq \ldots \geq b_n$ such that $b_n<1$ it holds that the support of the center is the interval $[b_n,1]$.
Additionally, Each peripheral seller $i$ has support $S_i=[b_i,b_{i-1}]$.
\end{corollary}
Note that in particular, there is no equilibrium with infinite-boundary for any of the CDFs.

\begin{lemma}
\label{lem:star-sketch}
For a star network with generic $\vec{\alpha}$, any equilibrium has finite boundary. Moreover, there is a unique sketch that can be satisfied in equilibrium.
\end{lemma}
\begin{proof}
We continue by presenting additional properties that must hold in any equilibrium.
The support of the center is the interval $[b_n,1]$.
For every $i$ with $b_i<1$, $b_i$ is in the support of the center, thus $u_0=u_0(b_i)=b_i(\alpha_0+i)$.
We conclude that
\begin{equation}
\label{eq.b.formula}
b_i= \frac{u_0}{\alpha_0+i}
\end{equation}
This means that if $b_i<1$ then $\frac{b_{i}}{b_{i-1}} = \frac{\alpha_0+i-1}{\alpha_0+i}=1-\frac{1}{\alpha_0+i}$.
Thus, once we fix some $j$ such that $b_{j-1}=1$ and $b_{j}<1$ we fix every $b_i$.

We next compute $\Fb_0(b_i)$ for every $i$ such that $b_i<1$, starting from $\Fb_0(b_n)=1$ and decreasing $i$ by one at every step.
For every peripheral seller $i$ with $b_i<1$ and every $x\in S_i$ 
it holds that $u_i(x)=x(\alpha_i+\Fb_0(x))$.
This holds in particular at $b_{i-1},b_i\in S_i$.
$$u_i=u_i(b_{i-1})=b_{i-1}(\alpha_i+\Fb_0(b_{i-1})) =
b_{i}(\alpha_i+\Fb_0(b_{i})) = u_i(b_{i})
$$
alternatively
$$ \Fb_0(b_{i-1})
=  \frac{b_{i}}{b_{i-1}}\cdot (\alpha_i+\Fb_0(b_{i})) -\alpha_i =
\left(1-\frac{1}{\alpha_0+i}\right)\cdot (\alpha_i+\Fb_0(b_{i})) - \alpha_i
$$
Thus
\begin{equation}
\label{eq.F.recurrence}
\Fb_0(b_{i-1})
= \Fb_0(b_{i})  - \frac{\alpha_i+\Fb_0(b_{i})}{\alpha_0+i}
\end{equation}

Equation \eqref{eq.F.recurrence} gives a recurrence for computing $\Fb_0(b_{i-1})$ given $\Fb_0(b_i)$, starting with $\Fb_0(b_n) = 1$,
this recurrence must hold in any equilibrium.
For generic $\vec{\alpha}$ it holds that $\Fb_0(b_i)\neq 0$ for every $i$.
In equilibrium it must be the case that $\Fb_0(b_i)\geq 0$.

{\bf Case 1:} If $\Fb_0(b_0)=\Fb_0(1)>0$ it means that the center must have an atom at $1$. This means that no other seller has any atom.
This implies that $u_0=\alpha_0$, which means that for every $i\in [n]$, $b_i= \frac{\alpha_0}{\alpha_0+i}$.
Thus for this case we have a unique sketch that can be satisfied in equilibrium.

We remark that for this case to happen it is necessary that $\alpha_0$ is quite large since
for every $i$ it must hold that $u_i = u_i(b_i) = b_i(\alpha_i + \Fb_0(b_i))\geq \alpha_i$ which implies that
$\frac{\alpha_0}{\alpha_0+i}(\alpha_i + 1)\geq \alpha_i$ and thus
$\alpha_0>i\alpha_i$.

{\bf Case 2:} If by using the recurrence of Equation \eqref{eq.F.recurrence} we get $\Fb_0(b_0)=\Fb_0(1)<0$, then let $j$ be the maximum (i.e.\ first) value for which \eqref{eq.F.recurrence} yields $\Fb_0(b_{j-1}) <0$ (thus $\Fb_0(b_{j}) > 0$ for a generic $\alpha$). It must hold that $j$ has an atom at $1$ and $b_{j-1}=1$ while $b_{j}<1$. This imply that $1=b_0=b_1=\ldots = b_{j-1}> b_j$ and thus every seller $i<j$ always price at $1$ (has an atom at $1$ of size $A_i(1)=1$) and no other seller has any atom.
Additionally, for any $i\geq j$ this allows us to fix every $b_i$ using the recursion $\frac{b_{i}}{b_{i-1}} = 1-\frac{1}{\alpha_0+i}$.
Thus for this case we have a unique sketch that can be satisfied in equilibrium.
\end{proof}

\subsection{Equilibrium Utilities}
\label{app:star-util}

As there is unique equilibrium in each star network, it is meaningful to talk about the equilibrium utilities of the seller.
We next aim to understand how the utilities of the sellers change as the sizes of the captive markets change slightly (change that is small enough such that the order of captive market sizes and the sketch of the equilibrium do not change).
We focus on the case that the center has a large market, large enough for the equilibrium to be of the first kind, with the center having an atom at $1$.

The center seller has an atom at $1$ and has utility $\alpha_0$.
For the peripheral sellers we can compute their utilities as follows.
Given $u_0=\alpha_0$ we can compute the utilities recursively, starting with $u_n$ and moving down to $u_1$.
It holds that
$$u_n=u_n(b_n)= b_n(\alpha_n+1) =  \frac{u_0}{\alpha_0+n}(\alpha_n+1) $$

Consider any peripheral seller $i-1$ that is not always pricing at $1$.
Since $b_{i-1}$ belongs to both $S_{i-1}$ and $S_i$, it holds that $u_i=u_i(b_{i-1})=b_{i-1}(\alpha_i+1-F_0(b_{i-1})) $
and $u_{i-1}=u_{i-1}(b_{i-1})=b_{i-1}(\alpha_{i-1}+1-F_0(b_{i-1})) $, thus
$$u_{i-1}=u_i + b_{i-1} ( \alpha_{i-1} -\alpha_{i} )  =
u_i + \frac{u_0}{\alpha_0+i-1} ( \alpha_{i-1} -\alpha_{i} )
=u_0\left( \frac{\alpha_n+1}{\alpha_0+n} +  \sum_{j=i}^{n} \frac{ \alpha_{j-1} -\alpha_{j} }{\alpha_0+j-1}\right)
$$

As we consider parameters for which case 1 holds, that is, the center has an atom at 1, and no other seller has an atom, we know that 
$u_0=\alpha_0$. We conclude that the center gains nothing from having access to additional market. Note that every peripheral seller $i$ has utility larger than $\alpha_i$ (as his equilibrium utility is it least as much as he can gain by pricing arbitrary close to 1).
The utility $u_i$ of each peripheral seller $i$ depends on $\alpha_0$ and the captive markets sizes $\alpha_i,\alpha_{i+1},\ldots, \alpha_n$, but not on the other captive market sizes. The utility $u_i$ increases in $\alpha_0$, but not linearly.
The dependence on every $\alpha_i,\alpha_{i+1},\ldots, \alpha_n$ is linear, 
increasing linearly with $\alpha_i$ and decreasing linearly in $\alpha_l$ in $\alpha_{i+1},\alpha_{i+2},\ldots, \alpha_n$.
To see this, observe that the linear coefficient of $\alpha_l$ is $\alpha_0\left(\frac{1}{\alpha_0+l} - \frac{1}{\alpha_0+l-1}\right)<0$.

\subsection{Non-uniqueness of Equilibrium: Lines with Captive Markets}

We now privide the details for the example that tree networks can exhibit multiple, non-utility-equivalent equilibria when there is 
more than one captive market, as depicted in figure \ref{fig:cex}(b) .  

Our network is a line of sellers $\{1, \dotsc, 6\}$, with weights $(\alpha_1, \dotsc, \alpha_6) = (10,1,1,1,1,10)$ and $(\beta_{12}, \beta_{23}, \beta_{34}, \beta_{45}, \beta_{56}) = (0.5,1,1,1,0.5)$.  

We now describe a sketch for this network.  Our set of boundary points is $1 = t_1 > t_2 > \dotsc > t_6$.
Sellers $1$ and $6$ have atoms at $1$.  The sellers' supports are $S_1 = [t_3, t_1]$, $S_2 = [t_6, t_5] \cup [t_3, t_1]$, $S_3 = [t_6, t_4]$, $S_4 = [t_5, t_2]$, $S_5 = [t_4, t_2]$ and $S_6 = [t_2, t_1]$.  
Treating program (LP1) from Section \ref{sec:sketches} as a quadratic program in which the boundary points $t_i$ are variables, we can obtain a sketch solution.  The boundary points at this solution are (approximately)
\[ (t_1, t_2, t_3, t_4, t_5, t_6) = (1, 0.960749, 0.960197, 0.87288, 0.610929, 0.576118) \]
and the relevant values of $\Fb_i(t_j)$ are given by
\[ \Fb_1(1) = 0.880591 \quad\quad \Fb_2(t_5) = 0.82906 \quad\quad \Fb_3(t_5) = 0.85755 \]
\[ \Fb_4(t_4) = 0.150999 \quad\quad \Fb_5(t_2) = 0.817081 \quad\quad \Fb_6(1) = 0.882248 \]
At this equilibrium, $u_2 = t_3( \alpha_2 + \beta_{1 2} ) = 1.4403$ and $u_5 = t_2( \alpha_5 + \beta_{5 6} ) = 1.44112$.

By symmetry of the network, there exists a second equilibrium with the order of the sellers reversed; in this alternative equilibrium, we have 
$u_2 = 1.44112$ and $u_5 = 1.4403$.  These two equilibria are therefore not utility-equivalent for the buyers, as required.




\end{document}